\pdfoutput=1
\documentclass[11pt]{article}
\usepackage{mlmodern}

\usepackage[T1]{fontenc}

\usepackage{fontawesome}

\usepackage[letterpaper,margin=1in]{geometry}

\usepackage[dvipsnames,table,xcdraw]{xcolor}
\usepackage[utf8]{inputenc}

\usepackage[
backend=biber,
style=alphabetic,
maxbibnames=15,
maxalphanames=10,
minalphanames=6,
doi=false,
isbn=false,
url=false,
eprint=false,
backref=true,
]{biblatex}
\DefineBibliographyStrings{english}{
  backrefpage  = {},
  backrefpages = {},
}

\DeclareFieldFormat{bracketswithperiod}{\mkbibbrackets{#1}}

\renewbibmacro*{pageref}{%
  \iflistundef{pageref}
    {}
    {\printtext[bracketswithperiod]{%
       \ifnumgreater{\value{pageref}}{1}
         {\bibstring{backrefpages}}
         {\bibstring{backrefpage}}%
       \printlist[pageref][-\value{listtotal}]{pageref}}%
     }
 }
     
\renewbibmacro{in:}{}
\AtEveryBibitem{\clearfield{pages}}

\definecolor{lb}{RGB}{0, 100, 200}
\definecolor{green2}{RGB}{60, 120, 0}

\usepackage[colorlinks,citecolor=green2,linkcolor=lb,bookmarks=true]{hyperref}

\usepackage[T1]{fontenc}    %
\usepackage{url}      %
\usepackage{booktabs}     %
\usepackage{nicefrac}     %

\usepackage{amsthm,amsfonts,amsmath,amssymb,epsfig,color,float,graphicx,verbatim,
enumitem}

\usepackage{algorithm,algorithmicx}
\usepackage[noend]{algpseudocode}

\usepackage{bbm}
\usepackage{caption}

\newlist{itemizec}{itemize}{2}
\setlist[itemizec,1]{label=\faCaretRight ,wide, parsep= 0.05pt, left = 15pt}

\usepackage[nameinlink,capitalise]{cleveref}

\usepackage{thm-restate}

\usepackage{thm-restate}
\usepackage{enumitem}
\usepackage{algorithm,algorithmicx}
\usepackage[noend]{algpseudocode}

\def\E{\mathbb E}

\def\R{\mathbb R}

\def\N{\mathbb N}

\def\spike{\eta}
\newcommand{\eps}{\epsilon}

\newcommand\numberthis{\addtocounter{equation}{1}\tag{\theequation}}

\let\vec\mathbf

\newcommand{\bA}{\vec{A}}
\newcommand{\bB}{\vec{B}}

\newcommand{\bI}{\vec{I}}

\newcommand{\bSigma}{\vec{\Sigma}}

\newcommand{\bY}{\vec{Y}}

\newcommand{\poly}{\mathrm{poly}}
\newcommand{\polylog}{\mathrm{polylog}}
\newcommand{\cN}{\mathcal{N}}

\newcommand{\cT}{\mathcal{T}}

\newcommand{\cP}{\mathcal{P}}

\newcommand{\cS}{\mathcal{S}}
\newcommand{\cA}{\mathcal{A}}

\newcommand{\cM}{\mathcal{M}}

\newcommand{\cB}{\mathcal{B}}

\newcommand{\Hsparse}[1]{\cB_{\text{$#1$-sparse}}}

\newcommand{\offdiag}{{\mathrm{offdiag}}}
\newcommand{\diag}{{\mathrm{diag}}}
\newcommand{\corr}{{\mathrm{corr}}}

\newcommand{\trace}{\operatorname{tr}}

\newcommand{\Var}{\operatorname{Var}}

\newcommand{\Ber}{\mathrm{Ber}}
\newcommand{\op}{\mathrm{op}}
\newcommand{\fr}{\mathrm{Fr}}

\newcommand{\proj}{\operatorname{Proj}}

\crefformat{equation}{(#2#1#3)}
\crefname{equation}{Equation}{Equations}
\crefname{lemma}{Lemma}{Lemmata}
\crefname{claim}{Claim}{Claims}
\crefname{fact}{Fact}{Facts}
\crefname{theorem}{Theorem}{Theorems}
\crefname{proposition}{Proposition}{Propositions}
\crefname{corollary}{Corollary}{Corollaries}
\crefname{remark}{Remark}{Remarks}
\crefname{definition}{Definition}{Definitions}
\crefname{question}{Question}{Questions}
\crefname{condition}{Condition}{Conditions}
\crefname{figure}{Figure}{Figures}

\newtheorem{theorem}{Theorem}[section]
\newtheorem{lemma}[theorem]{Lemma}

\newtheorem{claim}[theorem]{Claim}
\newtheorem{proposition}[theorem]{Proposition}

\theoremstyle{definition}

\newtheorem{definition}[theorem]{Definition}

\newtheorem{question}[theorem]{Question}
\newtheorem{problem}[theorem]{Problem}

\newtheorem{condition}[theorem]{Condition}

\allowdisplaybreaks

\theoremstyle{definition}
\newcommand{\sign}{\text{sgn}}

\usepackage{color}
\definecolor{Red}{rgb}{1,0,0}
\definecolor{Blue}{rgb}{0,0,1}
\definecolor{DGreen}{rgb}{0,0.55,0}
\definecolor{Purple}{rgb}{.75,0,.25}
\definecolor{Grey}{rgb}{.5,.5,.5}

\title{A Sub-Quadratic Time Algorithm for \\Robust Sparse Mean Estimation}

\author{
Ankit Pensia
\\
IBM Research \\
{\tt ankitp@ibm.com}\\
}

\begin{document}
\maketitle

\begin{abstract}
We study the algorithmic problem of sparse mean estimation in the presence of adversarial outliers. Specifically, the algorithm observes a \emph{corrupted} set of samples from $\mathcal{N}(\mu,\mathbf{I}_d)$,  where the unknown mean $\mu \in \mathbb{R}^d$ is constrained to be $k$-sparse.
A series of prior works has developed efficient algorithms for robust sparse mean estimation with sample complexity $\mathrm{poly}(k,\log d, 1/\epsilon)$ and runtime $d^2 \mathrm{poly}(k,\log d,1/\epsilon)$, where $\epsilon$ is the fraction of contamination.
In particular, the fastest runtime of existing algorithms is quadratic ($\Omega(d^2)$), which can be prohibitive in high dimensions.
This quadratic barrier in the runtime stems from the reliance of these algorithms on the sample covariance matrix, which is of size $d^2$.
Our main contribution is an algorithm for robust sparse mean estimation which runs in \emph{subquadratic}
time using $\mathrm{poly}(k,\log d,1/\epsilon)$ samples.
We also provide analogous results for robust sparse PCA. 
Our results build on algorithmic advances in detecting weak correlations, a generalized version of the light-bulb problem by Valiant~\cite{Valiant15}.    
\end{abstract}
\section{Introduction} %
\label{sec:introduction}

Mean estimation, a
fundamental unsupervised inference task studied in literature, may be described as follows: Given a family of distributions $\cP$ over $\R^d$,
the algorithm observes a set of i.i.d.\ points from an unknown $P \in \cP$, with the goal of outputting $\widehat{\mu}$ such that, with high probability, $\|\widehat{\mu} - \mu\|_2$ is small.
Although this framework is well-studied in the literature,
the data observed in practice may 
deviate from
the i.i.d.\ assumption and additionally may contain outliers.
Crucially, these outliers can easily break standard off-the-shelf estimators, for example, sample mean, geometric median, and coordinate-wise median.
To address this challenge, the field of robust statistics was initiated in the 1960s, aiming to develop algorithms that are robust to outliers~\cite{Hub64,AndBHHRT72,HubRon09}.
Before proceeding further, we formally define the contamination model we study in this paper.
\begin{definition}[Strong Contamination Model]
\label{def:contamination}
Given a \emph{corruption} parameter $\eps \in (0,1/2)$ 
and a distribution $P$ on uncorrupted samples, 
an algorithm obtains samples from $P$ with \emph{$\eps$-contamination} 
as follows: 
(i) The algorithm specifies the number $n$ of samples it requires. 
(ii) A set $S$ of $n$ i.i.d.\ samples from $P$ is drawn but not yet shown to the algorithm. 
(iii) An arbitrarily powerful adversary then inspects $S$, 
before deciding to replace any subset of $\lceil \eps n \rceil$ 
samples with arbitrarily corrupted points (``outliers'') to obtain the contaminated set $T$, which is then returned to the algorithm.
We say  $T$ is an $\epsilon$-corrupted version of $S$ and a set of
$\epsilon$-corrupted samples from $P$.
\end{definition}

Our focus will be on high-dimensional distributions, i.e., when the distribution $P$ is over $\R^d$ for large $d$. Dealing with outliers becomes harder in high dimensions  
because \textit{classical} outlier screening procedures (which otherwise work well in low dimensions) rely on the norm of the data points and are too coarse to distinguish outliers from inliers. 
Nevertheless, a long line of research, spurred by advances in ~\cite{DiaKKLMS16-focs,LaiRV16}, has developed a systematic theory of handling outliers in high-dimensional 
robust statistics~\cite{DiaKan22-book}. Notwithstanding this progress, major gaps persist in our \emph{fine-grained} understanding of fast robust algorithms for data with additional structure. 

Structured high-dimensional data distributions are ubiquitous in practice, e.g., 
natural images and sounds. Moreover, leveraging these underlying structures often dramatically improves algorithmic performance, e.g., in terms of error. 
A well-studied structure both in the theory and practice of high-dimensional statistics is \emph{sparsity},
see, for example, the textbooks \cite{EldKut12, HasTW15,vandeGeer16}. 
Consequently, we concentrate our efforts on \emph{structured} mean estimation, where we assume that the underlying mean is sparse,
i.e., an overwhelming majority of its coordinates 
are zero. 

In light of the challenges posed by outliers above and the prevalence and importance of sparsity, we study the problem of \emph{robust} sparse mean estimation.
We say a vector $x \in \R^d$ is $k$-sparse if $x$ has at most $k$ non-zero entries.
Our focus is on the practically relevant regime where $k$ is much smaller than $d$, say, poly-logarithmic in $d$. We formally define robust sparse mean estimation below.
\begin{problem}[Gaussian Robust Sparse Mean Estimation]
\label{prob:gaussian-rme}
Let $\epsilon_0 \in (0,1/2)$ be a sufficiently small constant.
Given $\epsilon \in (0, \epsilon_0)$, sparsity $k \in \N$, and a set of $\eps$-corrupted set of samples from $\cN(\mu,\bI_d)$ 
with
an unknown $k$-sparse mean $\mu \in \R^d$, the goal is to output an estimate $\widehat{\mu} \in \R^d$ such that $\|\widehat{\mu} - \mu\|_2$ is small with high probability.    
\end{problem}

Robust sparse mean estimation algorithms, efficient both in runtime and samples, were first developed in \cite{BalDLS17}, with sample complexity $n=\poly(k,\log d,1/\eps)$, runtime $\poly(d,n,1/\eps)$, and near-optimal error $\|\widehat{\mu} - \mu\|_2 = \tilde{O}(\eps)$.\footnote{In the  presence of outliers, vanishing error is usually not possible. In our setting, this is because it is impossible to distinguish two isotropic Gaussian distributions that are $\Omega(\eps)$-far apart in the presence of $\eps$-fraction of contamination.
}
In particular, the sample complexity is only poly-logarithmic in the ambient dimension $d$, thereby permitting statistical inference with far fewer samples than the $\Omega(d)$ samples  required by unstructured mean estimation.
Therefore,  for our algorithm for \cref{prob:gaussian-rme}, we set as the first requirement this sample complexity of $\poly(k, \log d, 1/\epsilon)$.

The focus of this work is to develop \emph{fast} robust sparse mean estimation with the aforementioned sample complexity.
Although the runtime of  \cite{BalDLS17} is polynomial in dimension, their algorithm uses the ellipsoid algorithm (which in turn solves a semidefinite program) and hence is not practical in high dimensions.
\cite{DiaKKPS19} then developed a \emph{spectral} algorithm with similar error guarantees and sample complexity and an improved runtime of $d^2 \poly(k,\log d,1/\eps)$.
Subsequent papers have proposed many algorithmic improvements and generalizations to a wider class of distributions~\cite{ZhuJS20,CheDKGGS21,DiaKKPP22-colt,DiaKLP22}; see \Cref{sec:related_work}.

\looseness=-1Despite this algorithmic progress, the fastest currently  known algorithm for \Cref{prob:gaussian-rme} is that of  \cite{DiaKKPS19} with runtime scaling as $d^2$.
This quadratic runtime of the algorithm can be prohibitive in high dimensions---the very  setting {that benefits most from sparsity} 
(because of sample-efficiency).
This quadratic dimension dependence is in stark contrast to the \emph{non-robust} setting (i.e., the outlier-free regime), where there exists a simple (folklore) algorithm\footnote{The algorithm computes the sample mean and thresholds  entries to ensure sparsity, hence failing
if there is even a single outlier.
Moreover, natural attempts to make this algorithm robust, such as coordinate-wise median, incur an highly suboptimal error of $\Omega(\eps \sqrt{k})$.
} with nearly-linear runtime, which is also minimax optimal.
This motivates 
the following fundamental question, highlighted in \cite{Diakonikolas19-stoctalk,Cheng21-ideal, Diakonikolas23-youtube}:
\begin{question}
\label{ques:linear}
\textit{Are there any nearly-linear time algorithms for robust sparse mean estimation?}
\end{question}
If we momentarily forgo sparsity (and the benefits that come along with it, e.g., the reduced sample complexity and interpretability) and focus on robust \emph{dense} estimation, then positive answers are known to \Cref{ques:linear},
see, e.g., \cite{CheDG19,DonHL19, DiaKKLT22-cluster, DiaKPP22-streaming}.
However, 
the sample complexities of the algorithms in these papers scale \textit{linearly} with dimension,\footnote{In fact, the overall runtime of these algorithms scales as $\tilde{\Theta}(nd) = \tilde{\Theta}(d^2/\eps^2)$, which is again quadratic in dimension.} 
which considerably exceeds our allowed budget of 
$\poly(k, \log d, 1/\eps)$ samples. 

\looseness=-1
In fact, as alluded to earlier, existing attempts at answering \cref{ques:linear} 
do not even break the \textit{quadratic} runtime barrier. This is due to natural technical obstacles within current algorithms:    
to robustly estimate the mean, 
they crucially
rely on the sample covariance matrix to detect outliers; 
but merely computing the sample covariance matrix costs $\Omega(d^2)$ time! 
Sparsity also precludes common tricks such as the power iteration to bypass explicitly writing the covariance matrix. 
Indeed, in certain parameter regimes,  even detecting atypical values of the covariance matrix from samples is conjectured to require $\Omega(d^2)$ time~\cite{DagSha18}. 
This begs the question whether this quadratic gap is inherent:
\begin{question}
\label{ques:subquadratic}
\textit{Is there an algorithm for robust sparse mean estimation that runs in $d^{2 - \Omega(1)}$ time and uses $\poly(k,\log d,1/\eps)$ samples?}
\end{question}

The main result of our work is an affirmative answer to \Cref{ques:subquadratic}. We hope our answer paves the path for progress towards answering \cref{ques:linear}, which was highlighted as an important open problem in \cite{Diakonikolas19-stoctalk,Cheng21-ideal, Diakonikolas23-youtube}. Our algorithm builds on advances in fast correlation detection algorithms by Valiant~\cite{Valiant15}.

\subsection{Our Results}

We establish the following result:
\begin{restatable}[Robust Sparse Mean Estimation in Subquadratic Time]{theorem}{ThmMainResult}
\label{thm:informal-main-result}
Let the contamination rate be $\epsilon \in (0,\epsilon_0)$ for a small constant $\epsilon_0 \in (0,1/2)$ and $k \in \N$ be the sparsity.
Let $T$ be an $\epsilon$-corrupted set of $n$ samples from $\cN(\mu,\bI_d)$ for an unknown $k$-sparse $\mu \in \R^d$.
Then there is a randomized algorithm $\cA$ that takes as input the corrupted set $T$, contamination rate $\epsilon$, sparsity $k \in \N$, 
and a parameter $q \in \N$, and produces an estimate $\widehat{\mu}$ such that
\begin{itemizec}
	\item (Sample Complexity and Error)  If $n \gtrsim (k^{2q} \log d)/\epsilon^{2q}$, then $\|\widehat{\mu} - \mu\|_{2} \lesssim \epsilon \sqrt{\log(1/\epsilon)}$ with probability at least $0.9$ over the randomness of the samples and the algorithm.\footnote{The success probability can be boosted to $1-\delta$ with a multiplicative increase of $\log(1/\delta)$ in the sample complexity and the runtime by repeating the procedure.}
	\item (Runtime) The algorithm runs in time at most
$
d^{1.62 + \frac{3}{q}} \poly(\log(d),k^q,1/\eps^q) \,.
$

\end{itemizec}
\end{restatable}
Several remarks are in order. The error guarantee of \Cref{thm:informal-main-result}, $O(\eps \sqrt{\log(1/\eps)})$, is nearly optimal even given infinite data and runtime.\footnote{The information-theoretic optimal error is $\Theta(\eps)$. Moreover, it is computationally hard to beat $\Theta\left(\eps \sqrt{\log(1/\eps)}\right)$ in the statistical query lower model~\cite{DiaKS17} and the low-degree polynomial tests~\cite{BreBHLS21} under \Cref{def:contamination}.} 
The main contribution of \Cref{thm:informal-main-result} is
the first algorithm for robust sparse mean estimation with runtime $d^{2 - \Omega(1)} \poly(k/\epsilon)$ and sample complexity $\poly(k,\log d, 1/\epsilon)$ (by selecting $q \in \N$ to be a constant bigger than $9$), thereby affirmatively answering \Cref{ques:subquadratic}.
As $q$ increases, the dependence of the runtime on the dimension
 approaches $d^{1.62}$. 
In particular, for a constant contamination rate $\epsilon$, we may set $q$ as large as $\gamma\left(\frac{\log d}{\log k}\right)$ for a  small $\gamma \in (0,1)$, and the algorithm retains sublinear (in $d$) sample complexity $d^{O(\gamma)}$ and subquadratic runtime $d^{1.62 + O(\gamma)}k^{O(1/\gamma)}$.
Finally, the sample complexity of \Cref{thm:informal-main-result} is (polynomially) larger than existing works; see \Cref{sec:discussion} for further remarks.

We next focus on robust version of  sparse principal component analysis (PCA).
Sparse PCA is a fundamental estimation task in high-dimensional statistics~\cite{dAspGJL07,HasTW15}, in which the algorithm observes samples from $\cN(0,\bI + \spike vv^\top)$ for a $k$-sparse unit vector $v\in \R^d$.
Similarly to sparse mean estimation, 
the standard algorithms for sparse PCA are not robust to outliers, and hence we study robust sparse PCA.
Again, existing robust sparse PCA algorithms from \cite{BalDLS17,DiaKKPS19,CheDKGGS21} require at least $\Omega(d^2)$ time.
In contrast, we establish the following result:

\begin{restatable}[Robust Sparse PCA in Subquadratic Time]{theorem}{ThmInformalRobustSparsePCA}
\label{thm:informal-robust-sparse-pca}
Let $T$ be an $\epsilon$-corrupted set of samples from $\cN(0, \bI_d + \spike vv^\top)$ for $\spike\in(0,1)$ and a $k$-sparse unit vector.
There is a randomized algorithm that takes as input the corrupted set $T$, contamination rate $\epsilon$, sparsity $k \in \N$, spike $\eta$, and a parameter $q \in \N$, and produces an estimate $\widehat{v}$ such that
\begin{itemizec}
    \item (Sample Complexity and Error) If $n \gtrsim \poly((k^q\log d)/ \epsilon^q)$, then $\|\widehat{v}\widehat{v}^\top - vv^\top\|_{\fr} \lesssim \sqrt{\epsilon \log(1/\epsilon)/\spike}$ with probability at least $1-\frac{1}{\poly(d)}$.
    \item (Runtime) The algorithm runs in time at most $d^{1.62 + \frac{3}{q}} \poly(n)$. 
\end{itemizec}
\end{restatable}
This result gives the first subquadratic time algorithm for dimension-independent error, improving on the $\Omega(d^2)$ runtime of \cite{BalDLS17,DiaKKPS19,CheDKGGS21}.
We note that the error guarantee of \Cref{thm:informal-robust-sparse-pca} is sub-optimal by a polynomial factor of $\eps/\spike$ (like in \cite{CheDKGGS21}), since the information-theoretic optimal error is $\eps/\spike$.
Despite this (polynomially) larger error, \Cref{thm:informal-robust-sparse-pca} is the first subquadratic time algorithm for robust sparse PCA with any non-trivial error, say, less than $0.01$.
We defer detailed discussion to \Cref{app:robust-sparse-pca}.

Our main technical ingredient in proving \Cref{thm:informal-main-result,thm:informal-robust-sparse-pca} is a result on detecting correlated coordinates of a high-dimensional distribution
by \cite{Valiant15}.
We give an overview of \Cref{thm:informal-main-result} in \Cref{sec:our-tech}, with 
 details in \Cref{sec:sparse_certificates_in_subquadratic_time},
 and defer \Cref{thm:informal-robust-sparse-pca} to \Cref{app:robust-sparse-pca}.

\subsection{Overview of Techniques}
\label{sec:our-tech}

We begin by presenting a brief overview of the landscape of current robust sparse mean estimation algorithms, followed by challenges in using these approaches to obtain an $o(d^2)$ runtime, and then conclude by presenting our algorithm.

\paragraph{(Dense) Robust Mean Estimation
} 
Let
$\mu'$ and $\bSigma'$ be the sample mean and the sample covariance of the current (corrupted) dataset.
The general guiding principle in robust \emph{dense} mean estimation is to use $\bSigma'$ to check for the presence of harmful outliers and iteratively remove them.
 Recall that inliers are sampled from an isotropic covariance distribution $\cN(\mu,\bI_d)$.
 Thus, if we take $\Theta(d/\epsilon^2)$ samples, then the variance of the inliers in any direction is $(1 \pm \tilde{O}(\eps))$. Moreover, the variance of any $(1 -\epsilon)$ fraction of inliers is $(1 \pm \tilde{O}(\eps))$.  

 The following are the key insights in developing algorithms for robust dense mean estimation: (i) Outliers cannot change the sample mean $\mu'$ of the data in any direction $v$ without significantly increasing the covariance $\bSigma'$ in the  direction $v$, (ii) Given a direction of large variance $v$ of the data (i.e., with variance larger than $1 + \tilde{\Omega}(\eps)$), one can reliably remove outliers by projecting the data onto $v$ and thresholding appropriately, and (iii) In the dense setting, the directions of large variance correspond to leading eigenvectors of the covariance matrix $\bSigma'$, and further they can be computed efficiently (in nearly-linear time) using power iteration.
Thus, one can iteratively remove outliers as follows: compute (approximately) the leading eigenvalues and eigenvectors of the sample covariance matrix $\bSigma'$ and remove the samples that have large projections along the computed direction.

\paragraph{Adapting to Sparsity and Smaller Sample Complexity}
For robust \textit{sparse} mean estimation, one can adapt the above strategy by focusing  only on the sparse directions $v$. Indeed, (i) and (ii)
above are straight-forward and the resulting sample complexity of the algorithm is $k \log(d)/\eps^2$ since we require concentration of the mean and the covariance only along $k$-sparse directions.
However, the problem of computing the direction of the leading \emph{sparse} eigenvalues of a matrix, $\max_{v: \|v\|_2=1, \text{$k$-sparse}} v^\top \bSigma' v$, is computationally-hard in the worst case.
Inspired by the literature on sparse PCA~\cite{dAspGJL07}, \cite{BalDLS17} proposed the following convex relaxation\footnote{For a matrix $\bA \in \R^{d \times d}$, $\|\bA\|_1$ denotes the $\ell_1$-norm of $A$ when flattened as a $d^2$-dimensional vector}:
\begin{align}
\label{eq:psd-relaxation}
    \sup_{\{\bA: \bA \succeq 0, \trace(\bA) = 1, \|\bA\|_1 \leq k\}}  \langle \bA, \bSigma' - \bI_d \rangle\,.
\end{align}
Given such a feasible $\bA$ with value larger than $\tilde{\Omega}(\eps)$, one can remove outliers provided a larger sample complexity of $(k^2 \log d)/\eps^2$.\footnote{This larger sample complexity, $k^2$ versus $k$, is due to the stronger concentration required by the relaxation.}
Although the resulting algorithm is polynomial-time and the desired sample complexity $\poly(k,\log d, 1/\eps)$, the algorithm requires solving semidefinite programs (SDPs), for which the current algorithms require time superquadratic in dimension~\cite{JiaKLPS20}.

\paragraph{Spectral Algorithm of \cite{DiaKKPS19}}
To avoid solving {the SDPs
 from the preceding paragraph,
\cite{DiaKKPS19} considers a different (and, in a sense, weaker)
 relaxation of sparse eigenvalues
by relying on the distributional properties of Gaussian distributions.
Let $\Hsparse{k^2}$ be the set of all (sparse) matrices $\bB$ with Frobenius norm 1
 and at most $k^2$ non-zero entries. 
 Importantly, $\Hsparse{k^2}$ contains all
 $vv^\top$ for all $k$-sparse unit vectors $v$.  
Their starting point is the observation is that for any $\bB \in \Hsparse{k^2}$, we have  $\Var_{x \sim \cN(0,\bI_d)}(x^\top \bB x) = 2\|\bB\|_\fr^2 =   2$. 
Thus,
for 
a fixed $\bB \in \Hsparse{k^2}$, the empirical mean of $x^\top \bB x$ over any $1-\eps$ fraction of  inliers should be $\trace(\bB) \pm \tilde{O}(\eps)$.
Moreover, standard uniform concentration arguments imply that this holds uniformly over $\bB \in \Hsparse{k^2}$ given $(k^2\log d)/\eps^2$ samples.
Their key observation is that the resulting (non-convex) optimization problem 
\begin{align}
\label{eq:frob-relaxation}
    \max_{\bB \in \Hsparse{k^2}} \langle \bB, \bSigma' -\bI_d \rangle   
\end{align}
can be solved via standard matrix operations (despite being non-convex) without resorting to SDPs (as opposed to \Cref{eq:psd-relaxation}); indeed, the optimal $\bB$ corresponds to the top-$k^2$ values of $\bSigma' - \bI_d$ in magnitude, computable in $\tilde{O}(d^2)$ time given $\bSigma'$.

Given such a feasible $\bB \in \Hsparse{k^2}$ that achieves $ \langle \bB, \bSigma'- \bI_d\rangle \ge \tilde{\Omega}(\eps)$, \cite{DiaKKPS19} also propose an efficient outlier-removal strategy. 
Overall, this yields 
an algorithm with runtime $d^2\poly(k/\eps)$ and sample complexity $(k^2 \log d)/\eps^2$. While a significant improvement over \cite{BalDLS17}, this still 
unfortunately falls short of our target runtime of $o(d^2)$.

The main challenge in extending \cite{DiaKKPS19}'s algorithm to get an $o(d^2)$ runtime is that one needs to write down $\bSigma'$ explicitly, which itself takes $\Omega(d^2)$ time.
Moreover, there is no known analog of power iteration for sparse settings with provable guarantees (recall that in the dense setting, the power iteration can be implemented in nearly-linear runtime~\cite{SacVis14}).
Our main technical insight is to use  advances in fast algorithms for correlation detection initiated by \cite{Valiant15}; 
see \Cref{sec:detecting_correlation_in_subquadratic_time} for precise statements.
Next, we explain why \cite{Valiant15} is potentially useful in our setting, explain the challenges in a direction application of their result, and our proposed fix.

\paragraph{Fast Correlation Detection Algorithm To The Rescue} Denote the correlation detection algorithm in \cite{Valiant15} by $\cA_{\text{corr}}$. 
In our setting, 
this algorithm 
guarantees
 that given a $\rho \in (0,1)$ and a large $q \in \N$, 
it
  can find (off-diagonal) coordinate pairs $(i,j)$ such that $|\bSigma'_{i,j}|\geq \rho$ in subquadratic time \emph{as long as} there are at most $O(d)$ off-diagonal coordinate pairs $(i,',j')$ such that $|\bSigma'_{i',j'}| \geq \rho^{q}$; observe that $\rho^q \ll \rho$. 

We now illustrate why such a subroutine
may
 be useful.
Suppose that $\cA_{\text{corr}}$ returns many correlated coordinate
pairs. Then the optimal value in \Cref{eq:frob-relaxation} must be large (if we optimize over $\Hsparse{k'}$ for some $k'$ large enough\footnote{Recall that we can interpret the maximum in \Cref{eq:frob-relaxation} as Euclidean norm of the largest $k^2$ entries, which would be at least $\rho \sqrt{k^2} = \rho k$.}),
and we can use those coordinates pairs
 to construct a $k'$-sparse $\bB$ that can be used to remove outliers as before \cite{DiaKKPS19}.
If,  on the other hand, $\cA_{\text{corr}}$ returns only a few coordinate pairs, then we know that only these small set of coordinates are (potentially) corrupted, and we have reduced our problem to robust dense mean estimation on these coordinates; recall that in this setting,  the sample mean is a good candidate for the remaining coordinates.
Thus, a fast correlation detection algorithm leads to a subquadratic time algorithm to filter outliers (or declare victory) provided that there are only $O(d)$ coordinate pairs with correlation larger than $\rho^q$. 

\paragraph{
Challenges in Applying Fast Correlation Detection and A Proposed Fix
}
A priori, it is unclear why there must be only $O(d)$ correlated coordinate pairs: indeed, the outliers are allowed to be dense (similar to inliers --- recall that only the population mean is 
 sparse), and, in the worst case,
 it is possible that they cause \textit{all} coordinate pairs to be correlated (on the corrupted data).
Thus, we need an alternative procedure to (i) detect if there are many $\rho^q$-correlated pairs and (ii) if so, find an alternative procedure to make progress. 

Fortunately, it turns out that
if there are $\Omega(d)$ correlated pairs, then a random pair 
has $\Omega(d/d^2)=\Omega(1/d)$ probability to be correlated.
Hence,
we can sample a relatively large --- but subquadratic, say $\Theta(d^{1.5})$ --- number of  random pairs
to estimate the true count of correlated pairs.
If random sampling does not find many such
pairs, then the true count would anyway  have been small with high probability, and we 
may
safely invoke \cite{Valiant15}'s algorithm, solving the detection problem (i) above.
On the other hand, if we do observe many (i.e.,
scaling polynomially with $d$) $\rho^q$-correlated pairs, then we know that the Frobenius norm of the largest $k' = \poly(k)$ entries of $\bSigma'-\bI_d$ must be large enough, $\Omega(\rho^q \sqrt{k'})$.
In other words, we can find a relatively sparse $\bB' \in \Hsparse{k'}$ such that $\langle \bSigma' - \bI_d, \bB'\rangle\geq\tilde{\Omega}(\eps)$.
Thus, we can iteratively remove outliers (or declare victory when safe to do so) irrespective of the number of correlated pairs.
Finally, the larger sample complexity of algorithm comes from invoking the filter on $k'$-sparse matrices $\bB$ for $k'\gg k$.
We give a more detailed overview in \Cref{sec:sparse_certificates_in_subquadratic_time}.

\subsection{Related Work}
\label{sec:related_work}

Our
work is situated within the field of algorithmic robust statistics, and we refer the reader to \cite{DiaKan22-book} for an extensive exposition on the topic.

\paragraph{Robust Sparse Estimation}
Efficient algorithms for robust sparse mean estimation were first developed in \cite{BalDLS17}, giving an algorithm to compute $\widehat{\mu}$ with sample complexity $\poly(k,\log d,1/\eps)$, runtime $\poly(d,1/\eps)$, and near-optimal error $\|\widehat{\mu} - \mu\|_2 = \tilde{O}(\eps)$.
Their algorithm used the ellipsoid algorithm with a separation oracle that requires solving an SDP. Invoking the ellipsoid algorithm can be avoided by using \cite{ZhuJS20} or through iterative filtering, but the resulting algorithm still requires solving multiple SDPs, which, as noted earlier, is inherently slow. Bypassing the use of SDPs,
\cite{DiaKKPS19} developed the first spectral algorithm for robust sparse estimation with runtime $d^2 \poly(k,\log d,1/\eps)$. Another novel take on this problem was seen in
\cite{CheDKGGS21}, which  gave an optimization-based algorithm
showing that first-order stationary points of a natural non-convex objective suffice.
Although the resulting algorithm relies on simple matrix operations, the derived runtime is super-quadratic in dimension.
In a different direction, \cite{DiaKKPP22-colt} and \cite{DiaKLP22} developed algorithms for robust sparse mean estimation for a wider class of distributions: heavy-tailed distributions and light-tailed distributions with unknown covariance, respectively. 

Robust sparse mean estimation is conjectured to have information-computation gaps~\cite{DiaKS17,BreBre20,DiaKKPP22-colt}. Particularly, while there exist inefficient ($d^k$-time) algorithms for robust sparse mean estimation using $(k \log (d))/\eps^2$ samples, all polynomial time algorithms are conjectured to require  $\Omega\left(k^2/\eps^2\right)$ samples~\cite{BreBre20}.

\paragraph{Fast Algorithms for Robust Estimation}
Looking beyond polynomial runtime as the criterion
of computational efficiency, a recent line of work has investigated \emph{faster} algorithms for a variety of robust estimation tasks:
 mean estimation~\cite{CheDG19,DonHL19,DepLec19,LeiLVZ20,DiaKKLT22-cluster,DiaKPP22-streaming,DiaKPP23-huber}, covariance estimation~\cite{CheDGW19}, principal component analysis~\cite{JamLT20,DiaKPP23-pca}, list-decoding~\cite{CheMY20,DiaKKLT22-cluster}, and  linear regression~\cite{CheATJFB20,DiaKPP23-huber}.
 The overarching goal in this line of work is to develop robust algorithms that have runtimes matching the corresponding non-robust off-the-shelf algorithms, thus reducing the computational overhead of robustness.
However, none of these algorithms is tailored to sparsity and hence have sample complexity scaling (at least) linearly with the dimension.\footnote{As a result, the dependence on the runtime again becomes quadratic because $nd = \Omega(d^2)$.}
\textit{
Our main contribution
is the first subquadratic runtime algorithm for robust sparse mean estimation with sample complexity $\poly(k \log d, 1/\eps)$.}

\paragraph{Fast Correlation Detection}
Given $n$ vectors in $\{\pm1\}^d$ and two thresholds $1 \geq \rho> \tau > 0$, the correlation detection problem 
asks to find
all coordinate pairs that have correlation at least $\rho$ given that not too many pairs have correlation 
larger
than $\tau$.
This problem is a generalization of the light bulb problem~\cite{Val88-light-bulb}.
The first subquadratic algorithm for both these problems was
given by \cite{Valiant15}, with further improvements and simplifications in  
 \cite{KarKK18-correlation,KarKKC20,Alm19}.
 Further algorithmic improvements for this task would likely also improve our
  runtime guarantees in \Cref{thm:informal-main-result,thm:informal-robust-sparse-pca}.

\section{Preliminaries}
\label{sec:preliminaries_for_sparse_estimation}

\paragraph{Notation} 
For a random variable $X$, $\E[X]$ denotes its expectation.
For a finite set $S$ and a function $g:S \to \R^d$, we use $\E_S[g(X)]$ to denote $(\sum_{x \in S}g(X))/|S|$.
We use $\poly(\cdots)$ to denote an expression that is polynomial in its arguments. The notations $\lesssim,\gtrsim, \asymp$ hide absolute constants.

For a vector $x\in\R^d$, we use $\|x\|_0$ and $\|x\|_2$
 to denote the  number of non-zero entries of $x$ and the $\ell_2$-norm of $x$, respectively.
For a vector $x$ and $k \in \N$, we define the $\|x\|_{2,k}$ norm as the maximum correlation between $x$ and a unit $k$-sparse vector, i.e., $\|x\|_{2,k} := \sup_{v: \|v\|_0 \leq k} \langle v, x\rangle \,$.
 Estimation in $\|\cdot\|_{2,k}$ immediately yields an estimate that is close in $\ell_2$ norm whenever $\mu$ is $k$-sparse:
\begin{proposition}[Sparse estimation using $\|\cdot\|_{2,k}$ norm~\cite{CheDKGGS21}]
\label{prop:sparse-estimation-using-2k-norm}
Let $x \in \R^d, y \in \R^d$, where $y$ is $k$-sparse.
Let $J\subset [d]$ be the top-$k$ coordinates of $x$ in magnitude, breaking ties arbitrarily. 
Define $x' \in \R^d$ to be $x'_i = x_i$ if $i \in J$ and $0$ otherwise.
Then $\|x' - y\|_2 \leq 6 \|x - y\|_{2,k}$. 
\end{proposition}
Thus, in the sequel, we solve the harder problem of estimating an \emph{arbitrary} mean $\mu\in \R^d$ in the $\|\cdot\|_{2,k}$ norm.

We denote matrices by bold capital letters, e.g., $\bA, \bSigma$.
We denote the $d\times  d$ identity matrix by $\bI_d$, omitting the subscript when clear.
For a matrix $\bA$, 
we use $\|\bA\|_0$,
and
$\|\bA\|_\fr$, 
to denote the number of non-zero entries 
and
the Frobenius norm,
respectively. 
For matrices $\bA$ and $\bB$ of the same dimensions, $\langle \bA, \bB \rangle$ to denotes the trace inner product $\trace(\bA^\top \bB)$.

For a subset $H \subset [d]$, and a vector $x \in \R^d$, define $(x)_H$ to be $|H|$-dimensional vector that corresponds to the coordinates in $H$.
 Similarly, for a matrix $\bA$, we define $(\bA)_H$ to be the $|H| \times |H|$ matrix corresponding to coordinates in $H$.
For a square matrix $\bA$, we use $\diag(\bA)$ and $\offdiag(\bA)$ to denote its diagonal and offdiagonal, respectively.

For a finite set $T \subset \R^d$, we define $\mu_T$ and $\bSigma_T$ to be the sample mean and the sample covariance of $T$, respectively.\footnote{Not to be confused with $(\bSigma)_H$ when $H \subset [d]$.}
When the set $T$ is clear from context, for a coordinate pair $(i,j) \in [d] \times [d]$ with $i \neq j$,  we denote the correlation between these coordinates on $T$ as $\corr(i,j) := \left| \bSigma'_{i,j}/\sqrt{\bSigma'_{i,i}\bSigma'_{j,j}} \right|$ for $\bSigma' = \bSigma_T$.
For a $\rho \in (0,1)$, we say coordinates $(i,j)$ are $\rho$-correlated if $\corr(i,j) \geq \rho$. 

Robust sparse estimation requires checking whether the current covariance matrix $\bSigma'$ has small quadratic forms, $v^\top (\bSigma' - \bI) v$, for sparse unit vectors.
For a matrix $\bA$ and $k \in \N$,
we define the sparse operator norm, $\|\bA\|_{\op,k} :=\sup_{v: \|v\|_2 =1, \|v\|_0 \leq k}|v^\top \bA v |$. 
Since computing $\|\cdot\|_{\op,k}$ is computationally hard, 
we look at the following relaxation from \cite{DiaKKPS19}:
For a matrix $\bA$, define $\|\bA\|_{\fr, k^2} :=\sup_{\bB: \|\bB\|_\fr = 1, \|\bB\|_0 \leq k^2}  \langle \bA, \bB\rangle\,$.
It can be seen that $\|\bA\|_{\op,k} \leq \|\bA\|_{\fr,k^2}$ since $\bB$ could be all $\pm vv^\top$ for $k$-sparse unit vectors $v$.
Moreover, $\|\bA\|_{\fr,k^2}$ is the Euclidean norm of the largest $k^2$ entries (in magnitude) of $\bA$.

Since we will routinely look at the projections of the points on a subset of coordinates, we formally define it below:
\begin{definition}[Projection of Pairs of Coordinates]
\label{def:set-projection}
Let $H_{\text{pair}} \subset [d] \times [d]$ be a set of pair of coordinates such that $(i,i) \not \in H_{\text{pair}}$ for any $i \in [d]$.
For an even $m \in [d^2]$, we define the operator $\proj_{m}$ that takes any such $H_{\text{pair}}$ and returns a set in $[d]$ that has a cardinality at most $m$  as follows:
\begin{itemizec}
	\item If $|H_{\text{pair}}| {\leq} \frac{m}{2}$, return $\{i: (i,j) \in H_{\text{pair}} \text{ or } (j,i) \in H_{\text{pair}}\}$.
	\item Otherwise, let any $m/2$ distinct elements of $H_{\text{pair}}$ be $(i_1,i_2)$, $\dots, (i_{m-1}, i_m)$. Return $\{i_j: j \in [m]\}$. 
\end{itemizec}
When the subscript $m$ is omitted, we take $m$ to be $d^2$.\end{definition}
Informally, the operator returns a set $H$ such that for any matrix $\bA$, for small $m$,  $ \|(\bA)_H\|_\fr^2 \geq \sum_{(i,j) \in H_{\text{pair}}}\bA_{i,j}^2$, while for larger $m$, $\|(\bA)_H\|_\fr^2 \geq m \min_{(i,j) \in H_{\text{pair}}} \bA_{i,j}^2$.

\subsection{Deterministic Condition on Inliers}
\label{sec:stability-prelim}
A recurring notion in developing robust algorithms is that of \textit{stability}, which stipulates that the first and second moment of the data not change much under removal of a small fraction of data points.
\begin{definition}[Stability]
\label{def:stability-inliers}
For an $\epsilon \in (0,1/2)$, $\delta \geq \epsilon$, and sparsity $k \in \N$, we say a set $S \subset \R^d$ is $(\epsilon, \delta,k)$-stable with respect to $\mu \in \R^d$ if the following holds for any subset $S' \subseteq S$ with $|S'| \geq (1 -\epsilon)|S|$:
	(i)  $\left\|\E_{S'}[X - \mu]\right\|_{2,k} \leq \delta$, and
	(ii)
	$\left\|\E_{S'}[(X - \mu)(X - \mu)^\top] - \bI_d\right\|_{\fr,k^2} {\leq} \delta^2/\epsilon$.
\end{definition}
The following result gives a nearly-tight bound on the sample complexity required to ensure stability.
\begin{lemma}[{Stability Sample Complexity~\cite[Lemma 3.3]{CheDKGGS21}}]
\label{lemma:stab-sample-complexity}
Let $S$ be a set of  $n$ i.i.d.\ samples from a subgaussian distribution $P$ over $\R^d$ such that $P$  has (i) mean $\mu \in \R^d$, (ii) identity covariance, and (iii) satisfies the
 Hanson-Wright inequality; in particular, $\cN(\mu,\bI_d)$ satisfies all three properties.
Then if  $n \gtrsim (k^2 (\log d)/ \eps^2)$, then  $S$ is $(\epsilon,\delta,k)$-stable with high probability with respect to $\mu$ for $\delta = C \eps \sqrt{\log(1/\eps)}$ where $C$ is a large absolute constant.
\end{lemma}
We note that the deterministic condition in \Cref{def:stability-inliers} is slightly stronger than \cite{CheDKGGS21}---the condition for the covariance---but their proof continues to work with the same sample complexity for \Cref{def:stability-inliers}.\footnote{We remark that \cite[Lemma 3.3]{CheDKGGS21} claims to establish the deterministic condition for all isotropic subgaussian distributions. However, their proof crucially uses Hanson-Wright inequality, which does not hold for arbitrary isotropic subgaussian distributions. In particular, the example in \cite[Exercise 6.3.6]{Vershynin18}, a uniform mixture of two spherical Gaussians, gives a counter example; see \cite[Section 6.3]{Vershynin18} for additional discussion.}

\subsection{Randomized Filtering}
\label{app:randomized-filtering}
We will use the following template of filtering algorithm from \cite[Section 2.4.2]{DiaKan22-book} (after a slight change in parameters). The following template 
for filtering has now become a standard in algorithmic robust statistics.
\begin{algorithm}[H]
\caption{Randomized Filtering}
\label{alg:randomized_filtering}
\begin{algorithmic}[1]
\State Let $T_1 \gets T$
\State $i \gets 1$
\While{$T_i \neq \emptyset $ and 
$T_i$ does not satisfy the stopping condition $\cS$
}
\State Get the scores $f: T_i \to \R_+$ satisfying $\sum_{x \in T_i \cap S} f(x) \leq \sum_{x \in T_i \setminus S} f(x)$ \label{line:score-f-generic} and $\max_{x \in T_i} f(x) > 0$
\State $T_{i+1} \gets T_i$
\For {each $x \in T_i$}
\label{line:sample-removal-prob}
\State Remove the point $x$ from $T_{i+1}$ with probability  $\frac{f(x)}{\max_{x \in T_i} f(x)}$
\EndFor
\State $i \gets i +1$
\EndWhile
\State \Return $T_{i}$
\end{algorithmic}
\end{algorithm}
These filtering algorithms have become a standard template in algorithmic robust statistics.
Here, the stopping condition $\cS$ can be a generic condition that can be evaluated in $\cT_{\text{stopping}}$ time using some algorithm $\cA_s$---it can be a randomized algorithm (using independent randomness from Line \ref{line:sample-removal-prob}) that succeeds with probability $1 - \tau$. 
We also require that whenever the stopping condition is not satisfied and the set $T_i$ is an $10 \eps$-corruption of $S$, then the scores $f:T_i \to \R_+$ satisfying the guarantees of Line \ref{line:score-f-generic} can be computed in time $\cT_{\text{score}}$. 

\begin{theorem}[Guarantee of \Cref{alg:randomized_filtering}; {\cite[Theorem 2.17]{DiaKan22-book}}]
\label{thm:generic-randomized-alg}
If the above stopping conditions and filter conditions are met, then \Cref{alg:randomized_filtering} returns a set $T' \subseteq T$ such that, with probability at least $8/9 - \tau|T|$, the following statements hold.
\begin{itemizec}
	\item \Cref{alg:randomized_filtering} runs in time $O(|T| (\cT_{\text{stopping}} + \cT_{\text{score}} + |T|))$.
	\item Each set $T_i \subseteq T$ observed throughout the run of the algorithm (in particular, the output set $T'$) 
 is a $10 \epsilon$-corruption of $S$.
\item $T'$ satisfies the stopping condition $\cS$.
\end{itemizec}
\end{theorem}

\subsection{Certificate Lemma and 
 Frobenius Norm Filtering}
\label{sec:certificate-lemma-prelim}

The following standard certificate lemma guides the algorithmic design: if the sample covariance matrix of the corrupted data is roughly isotropic, then the sample mean is a good estimate.
\begin{lemma}[{Sparse Certificate Lemma, see, e.g., \cite{BalDLS17}}]
\label{lem:certificate-lemma}
Let $T$ be an $\epsilon$-corrupted version of $S$, where $S$ is $(\epsilon,\delta,k)$-stable with respect to $\mu$ (cf.\ \Cref{def:stability-inliers}).
Then, 
$$\|\mu_T - \mu\|_{2,k} \lesssim \delta + \sqrt{ \epsilon \|\bSigma_T - \bI_d\|_{\op,k} } \,. $$
\end{lemma}

We now state the guarantee of filtering procedures,
where the goal is to filter outliers from a corrupted set $T$.
In dense mean estimation, the most common filters are based on scores of the form $(v^\top (x - \mu_T))^2$ for a direction of large variance $v$; this filter is guaranteed to succeed as long as the covariance matrix $\bSigma_T$ is far from the identity in operator norm.
In our setting, we will need a stronger filter guarantee that is guaranteed to succeed under the weaker condition that the covariance (submatrix) matrix is far from the (submatrix) identity in the Frobenius norm; Observe that the operator norm of (the corresponding submatrix of)  $\bSigma_{T} - \bI_d$ can be much smaller.
The following lemma
corresponds to the above situation and simplifies \cite[Steps 6-10 of Algorithm 1]{DiaKKPS19}:  
\begin{restatable}[Sparse Filtering Lemma]{lemma}{LemSparseFilteringLemma}
\label{lem:sparse-filtering-lemma}
Let $\epsilon \in (0 ,\epsilon_0)$ for a small absolute constant $\epsilon_0$.
Let $T$ be an $\epsilon$-corrupted version of $S$, where $S$ is $(\epsilon,\delta,k)$-stable with respect to $\mu$.
Let $H \subset [d]$ be such $\left\|(\bSigma_T - \bI)_H\right\|_\fr = \lambda$ for $\lambda \gtrsim \delta^2/\epsilon$ and $|H| \leq k$.
There exists an algorithm $\cA$ that takes $T$, $H$, $\epsilon$, and $\delta$ and returns scores $f:T \to \R_+$ so that $\sum_{x \in S \cap T} f(x) \leq \sum_{x \in T\setminus S } f(x)$, i.e., the sum of scores over inliers is less than that of outliers, and $\max_{x \in T}f(x) > 0$.
Moreover, the algorithm runs in time $d\cdot\poly(|H||T|)$.
\end{restatable}
These scores can be used to filter points from $T$ such that on expectation over the algorithm's randomness, more outliers are removed than inliers~\cite{DiaKan22-book}.
We give a proof of \Cref{lem:sparse-filtering-lemma} in \Cref{app:preliminaries}.

Thus, if $\bSigma_T - \bI_d$ has large (sparse) Frobenius norm, then we can make progress by removing outliers.
The contribution to this norm from the \emph{diagonal} entries can be calculated efficiently in $O(dn)$ time, and if large, then can also be used to remove outliers.
Thus, we will assume that the corrupted set has already been pre-processed to satisfy the following:
\begin{restatable}[Preprocessing]{condition}{CondDataProcess}
\label{cond:idealistic-condition}
Let $T$ be an $\epsilon$-corrupted version of $S$, where $S$ is $(\epsilon,\delta,k)$-stable.
Suppose $T$ satisfies $ \|\diag(\bSigma_T - \bI_d)\|_{\fr,k^2} \leq \min\left( O(\delta^2/\epsilon), 0.5 \right)$.
\end{restatable}
For completeness, we give details in \Cref{app:preprocessing}.
The next result, also proved in \Cref{app:preprocessing}, shows that any further small modifications of the preprocessed sets retains small sparse operator norm. 
\begin{restatable}{claim}{LemInheritStability}
\label{lem:inherit-stability-mean}
Let $C$ be a large enough constant $C>0$.
 Let $T'' \subset T'$ be two $O(\eps)$-contamination of $S$ such that $S$ is an $(C\eps,\delta,k)$-stable with respect to $\mu$.
 Suppose that $\|\diag(\bSigma_{T'} - \bI_d) \|_{\op,k} \lesssim \delta^2/\eps$. Then
 $\|\diag(\bSigma_{T'} - \bI_d) \|_{\op,k} \lesssim \delta^2/\eps$
\end{restatable}

\subsection{Detecting Correlation in Subquadratic Time}
\label{sec:detecting_correlation_in_subquadratic_time}

We will use \cite[Theorem 2.1]{Valiant15} that can detect $\rho$-correlated coordinates in subquadratic time if there are not too many $\tau$-correlated coordinates for $\tau \ll \rho$, say $\rho^{3}$.
\begin{theorem}[Fast Correlation Detection~\cite{Valiant15}]
\label{thm:robust-correlation-detection}
Let $\rho \in (0,1)$ be strong correlation threshold and $\tau \in (0,1)$ be margin threshold with $\rho > 12 \tau$.
Let $T$ be a set of $n$ vectors in  $\R^d$ such that
there are at most $s$ $\tau$-correlated coordinate pairs.
Then, there is an algorithm that takes $\rho, \tau, T$ as input, and, with probability $1-o(1)$, will output all $\rho$-coordinate pairs.
 Additionally, the runtime of the algorithm is 
$\left( sd^{0.62} + d^{1.62 + 2.4\frac{\log(4/\rho)}{\log(1/3\tau)}}\right) \poly(n, \log d,1/\tau)$.
\end{theorem}
The version above follows from \cite[Theorem 2.1]{Valiant15} for the binary vectors  using standard reductions, for example, \cite[Lemma 4.1]{Valiant15}.
For completeness, we state the guarantees of \cite[Theorem 2.1]{Valiant15} and show the reduction in \Cref{app:preliminaries}.

\section{Robust Sparse Mean Estimation in Subquadratic Time} %
\label{sec:sparse_certificates_in_subquadratic_time}
In this section, we explain our main technical contribution: a fast algorithm for robust sparse mean estimation under the stability condition.
\begin{restatable}[Robust Sparse Mean Estimation in Subquadratic Time]{theorem}{ThmMainResultFormal}
\label{thm:formal-main-result}
Let $c$ be a small enough absolute constant and $C$ be a large enough absolute constant.
Consider
 the corruption rate $\eps \in (0,\eps_0)$, where $\eps_0$ is a small enough absolute constant.
Let $k \in \N$ be the sparsity parameter and $q \in \N$  the correlation decay parameter with $q\geq 3$.
Let $T$ be an $\eps$-corrupted version of a set $S$, where $S$ satisfies $(C\eps,\delta,k')$-stability with respect to $\mu$ for $k':= \frac{(Ck)^q}{(\delta^2/\eps)^{q-1}}$ and $\delta^2/\eps \leq c$.
Then there is a randomized algorithm (\Cref{alg:wrapper-algorithm}) that takes as inputs $T$,  $\epsilon$, $\delta$, $k $, 
and  $q$ and produces an estimate $\widehat{\mu}$ such that, with a probability at least $1-1/d^2$ over the randomness of the algorithm, we have the following guarantees:
\begin{itemizec}
	 \item (Error) $\|\widehat{\mu} - \mu\|_{2,k} \lesssim \delta$.
	\item (Runtime) The algorithm runs in time at most
$d^{1.62 + \frac{3}{q}}\poly\left( |T|, \log d, k^q, 1/\eps^q \right)$.

\end{itemizec}
\end{restatable}

\paragraph{Organization} In \Cref{sec:alg-blueprint}, we highlight the key technical challenges of obtaining $o(d^2)$ runtime algorithm and our proposed fix. We record the guarantees of the key subroutines in \Cref{sec:sparse-certificate-key-procedures}. The proof of \cref{thm:formal-main-result} is given in \Cref{app:proof_of_cref_thm_formal_main_result}, and we finally show in \Cref{app:proof_of_cref_thm_informal_main_result}  how \Cref{thm:formal-main-result} implies \Cref{thm:informal-main-result}.

\subsection{Algorithmic Blueprint of \Cref{thm:formal-main-result}}

\label{sec:alg-blueprint}
To establish \Cref{thm:formal-main-result}, we start with the following blueprint for robust sparse mean estimation, with the aim of implementing it in $o(d^2)$ time.
\begin{algorithm}[H]
\caption{Algorithmic Blueprint}
\label{alg:outline}
\begin{algorithmic}[1]
	\State  Compute $\|\bSigma_T - \bI_d\|_{\fr,k^2}$ approximately. \label{item:blueprint-frob-norm}
	\While{$\|\bSigma_T - \bI_d\|_{\fr,k^2}$ is large}
		\State 
		Let $H$ be the corresponding coordinates with large Frobenius norm.
		\label{item:blueprint-H-matrix}
		\State
		 Filter $T$ using $H$ in  \Cref{lem:sparse-filtering-lemma}. 
	\EndWhile
	\State  Output the sample mean $\mu_T$. 
\end{algorithmic}
\end{algorithm}
The problem in implementing this blueprint naively is that Steps~\ref{item:blueprint-frob-norm} and \ref{item:blueprint-H-matrix} in \Cref{alg:outline} take $d^2$ time.
However, these are the only two bottlenecks:
Each filtering step takes only $d\poly(k,n) = d\poly(k/ \eps)$ time, and there are at most $n = \poly(k/\eps)$ iterations.
As we describe below, our goal in this section is to use \Cref{thm:robust-correlation-detection} to speed up the Steps~\ref{item:blueprint-frob-norm} and \ref{item:blueprint-H-matrix}.

\paragraph{Usefulness of Fast Correlation Detection}
\looseness=-1We will run \Cref{thm:robust-correlation-detection} to identify the off-diagonal indices $(i,j) \in [d] \times [d]$ that are correlated. 
Recall that \Cref{thm:robust-correlation-detection} takes two arguments $\rho$ (the threshold for strong correlation) and $\tau$ (the margin threshold).
Suppose we fix $\rho$  to be small, roughly $\frac{\delta^2}{\epsilon \poly(k)}$. Then \Cref{thm:robust-correlation-detection} reports back all $\rho$-correlated coordinate pairs. For the time to be subquadratic in $d$, we require the number of $\tau$-correlated pairs to be small. 
Let $H_{\text{pair}} \subset [d]\times [d]$ be the set of coordinate pairs in the output, and  define $H$ to be the set of all coordinates that appear in $H_{\text{pair}}$; $\proj(H_{\text{pair}})$ from \Cref{def:set-projection}  to be more formal.
Then, one of the following two cases must be true:
\begin{itemizec}
	\item $|H|$ is small ($\poly(k/\epsilon)$): 
	Since coordinates in $H^\complement$ have correlations at most $\rho$, then we know that 
	$\|(\bSigma_T - \bI)_{H^\complement}\|_{\fr,k^2} \leq k\rho$, which can be made less than $\delta^2/\epsilon$ for $\rho$ small enough ($\rho \lesssim \delta^2/(k\epsilon )$). Thus, the sample mean on $H^\complement$ is a good estimate in $\ell_{2,k}$ norm (cf.\ \Cref{lem:certificate-lemma}). On $H$, we can use a dense mean estimation algorithm, which would be fast as $|H| = \poly(k/\epsilon)$.

\item $|H|$ is large: Since there are many coordinate pairs with correlation at least $\rho$, we can filter and iterate as follows:
If we take any $H' \subset H$ of size $k'$, then each row and column in $(\bSigma_T - \bI)_{H'}$ has an entry larger than $\rho$ (in absolute value), and thus  $\|(\bSigma_T - \bI)_{H'}\|_{\fr} \geq \rho \sqrt{k'}$.
By taking $k'$ large enough (larger than $\delta^4/(\rho^2\epsilon^2)$), the resulting quantity will be bigger than $\delta^2/\epsilon$, allowing us to filter if the inliers satisfy stability with parameter $k'$ (cf.\ \Cref{lem:sparse-filtering-lemma}).
\end{itemizec}

Thus, we can implement Steps \ref{item:blueprint-frob-norm} and \ref{item:blueprint-H-matrix} in \Cref{alg:outline} fast using 
\Cref{thm:robust-correlation-detection}, so long as (i) inliers satisfy $(\epsilon,\delta,k')$ stability, (ii) we choose
$\rho^2 \asymp \frac{\delta^4}{\epsilon^2 k^2}$  and $k' \asymp \frac{\delta^4}{\epsilon^2 \rho^2} \asymp k^2$,
and (iii) \emph{there are not too many $\tau$-correlated pairs.}

\paragraph{Challenges in Using Fast Correlation Detection}
Suppose we set $\tau = \rho^q$ for some $q \geq 3$.\footnote{We choose this parameterization because the runtime of \Cref{thm:robust-correlation-detection} depends on $\log(1/\rho)/\log(1/\tau) = 1/q$.}
Looking at \Cref{thm:robust-correlation-detection}, we obtain a subquadratic time algorithm \emph{only} if $s$, the number of $\tau$-correlated pairs for $\tau := \rho^q$, is much smaller
than $d^{1.38}$; In fact, we will impose $s$ to be less than $d$ so that it is not the dominant factor in the runtime.
A priori, there is no reason for there to be
at most $d$ coordinate pairs (out of $d^2$ pairs) that are $\tau$-correlated.
Thus, we need a way to detect this situation
and find an alternative way to make progress.

\paragraph{Proposed Solutions: Efficient Detection and Filtering}
We begin with the detection procedure.
If there are $\Omega(d)$ many $\tau$-correlated pairs, then a pair
 sampled uniformly at random has a probability of $\Omega(d^{-1})$ of being $\tau$-correlated.
Thus, if we check many random coordinate pairs, superlinear but subquadratic, then we can accurately guess $s$.

In particular, let $U$ be the number of $\tau$-correlated pairs that were observed out of $m$ random pairs (sampled with replacement). 
Then, $U$ is distributed roughly as $\text{Ber}(m, s/d^2)$.
Binomial concentration implies that 
 $s \lesssim   (d^2U/m) + (\log d)(d^2/m)$ with probability at least $(1 - 1/d^2)$.
Taking $m$ to be $d^{1.5}$, we see that $s \lesssim \sqrt{d} U + \sqrt{d}\log d$. 
Thus, we obtain a fast (randomized) check to see if $s$ is less than $d$ than runs in $mn = d^{1.5} \poly(k/\eps)$ time: simply check if $U\leq \sqrt{d}$.
Thus, it remains to ensure that we can make progress when $U$ is large, in particular, $\Omega(\sqrt{d})$.

Let $H'_\text{pair}$ be the $\tau$-correlated coordinates that were observed in the above procedure; $U= |H'_\text{pair}|$. 
Crucially, we are in the regime when $|H'_\text{pair}| \geq \sqrt{d}$.
We want to use a small subset of the coordinates in $H'_\text{pair}$ to filter outliers. 
Let $H$ be $k''$-sized set of coordinates that appear in $H'_\text{pair}$; formally, $H := \proj_{k''}(H'_\text{pair})$.
Thus, each row and column in $(\bSigma_T - \bI)_{H}$ has an entry larger than $\tau$ in absolute value, implying that $\|(\bSigma_T - \bI)_{H}\|_\fr \geq \tau \sqrt{k''}$.
Therefore, we can use this $H$ to filter using \Cref{lem:sparse-filtering-lemma}
as long as the original set is also $(\epsilon,\delta,k'')$-stable and $ \sqrt{k''} \tau \gg \delta^2/\epsilon$, i.e., $k'' \asymp \frac{\delta^4}{ \epsilon^2 \tau^2} \asymp\frac{\delta^4}{ \epsilon^2 \rho^{2q}} \asymp \frac{\delta^{4}}{ \epsilon^2 (\delta^4 /k\epsilon^2 )^q} \asymp \frac{k^{2q}}{ (\delta^4/\epsilon^2)^{q-1}}$.

Observe that we require the inliers to satisfy $(\epsilon,\delta,k'')$- stability for $k'' = \poly(k^q/\epsilon^q)$. 
By \Cref{lemma:stab-sample-complexity}, a set of $(k'')^2/\epsilon^2$ many i.i.d.\ points will satisfy this stability condition, giving us the sample complexity.

\subsection{Sparse Certificates and Filters in Subquadratic Time}
\label{sec:sparse-certificate-key-procedures}
We now give formal guarantees of the key procedures outlined above.
First, consider the procedure
that randomly samples the coordinates to estimate the number of weakly-correlated coordinates.
\begin{restatable}{lemma}{LemRandomlyCheckCoordinates}
\label{lem:randomly-check-coordinates}
Let $T \subset \R^d$ be a multiset with covariance matrix $\Sigma'$.
Let $m \in N$ be the sampling parameter.
Let $J_* \subset [d] \times [d]$ be the off-diagonal coordinate pairs such that $|\bSigma'_{i,j}| \geq \tau$.
Then \Cref{alg:randomlycheck} takes $\tau$, $m$, and $T$ as input and returns a set $J \subset J_*$ in time $O(m|T|)$ such that with probability $ 1- 1/d^2$,
$|J| \geq \frac{m |J_*|}{4d^2}   - 16 \log d$. 
\end{restatable}

\begin{restatable}{algorithm}{AlgRandomlyCheck}
\caption{\textsc{RandomlyCheckCoordinates}}
\label{alg:randomlycheck}
\begin{algorithmic}[1]
\State
Let $H_{\text{pair}} \subset [d] \times [d]$ of size $m$, with $(i,j) \in H_{\text{pair}}$ sampled i.i.d.\ from off-diagonal elements.
 	\State Let $J  \gets \{ (i,j) \in H_{\text{pair}}: |\bSigma'_{i,j}| \geq \tau \}$ for $\bSigma'= \bSigma_T$ 
 	\State \Return $J$.\end{algorithmic}
\end{restatable}
\begin{proof}
To show correctness, we shall use the following concentration inequality for Binomials: If $X \sim \Ber(n,p)$, then with probability $1- \delta$, $|\sqrt{X} - \sqrt{np}| \leq 2\sqrt{\log(1/\delta)}$. See, for example, \cite[Equations (15.21) and (15.22)]{PolWu23}.
In particular, with probability $1-1/d^4$,
$ \sqrt{X} \geq \sqrt{np} - 4 \sqrt{\log d}$,
which implies $X \geq 0.25np - 16 \log d$.\footnote{We use that if $a,b,c \in \R_+$ then $a \geq b - c$ implies that $a^2 \geq b^2/4 - c^2$.
The proof is as follows: if $b \geq 2c$ then $ a \geq b/2$ and thus $a \geq b^2/4 \geq b^2/4 - c^2$; otherwise $ b^2/4 \leq c^2  $ and  thus $a^2 \geq 0 \geq b^2/4 - c^2$ holds trivially.}

The probability that a single pair in $H_{\text{pair}}$ has correlation of magnitude at least $\tau$ is exactly $(|J_*|/d(d-1))$, and  thus $|J| \sim \Ber(m, |J_*|/(d(d-1)))$.
Therefore, applying the Binomial concentration,
with probability $1- 1/d^2$,
it holds that $|J| \geq \frac{m |J_*|}{4d^2} - 16 \log d$.
The claim about the runtime is immediate.

\end{proof}

Combining 
\Cref{alg:randomlycheck}
with \Cref{thm:robust-correlation-detection}, 
\Cref{alg:win-win-algorithm-app}
either returns a small subset of coordinates with large Frobenius norm or indicates 
when all tiny subsets have small Frobenius norm.

\begin{restatable}[Subroutine to Identify Corrupted Coordinates \Cref{alg:win-win-algorithm-app}]{proposition}{PropWinWinGuarantee}
\label{prop:win-win-guarantee}
Suppose we are given a corrupted set $T$, Frobenius threshold $ \kappa \in \R$, correlation threshold $\rho \in (0,1)$, margin threshold $\tau \in (0,\rho/12)$, sampling parameter $m \in \N$, and correlation count $s \in \N$.
Suppose that each diagonal entry of $\bSigma_T$ lies in $[1/2,3/2]$ and
$\frac{c_1d^2}{m} \left( \frac{\kappa^2}{\tau^2} + \log d \right) \leq s$.

Then \Cref{alg:win-win-algorithm-app} takes as input $T, \kappa, \rho, \tau, m, s$ as input and satisfies the following with probability $1 - 1/\poly(d)$:
\begin{itemizec}
	\item[\faCaretRight] It either outputs a set  $H \subset [d]$ with $|H| \leq \frac{\kappa^2}{\tau^2}$ and $\| \left(\mathrm{offdiag}(\bSigma_T)\right)_H\|_{\fr}  \geq \kappa/2$.
	\item[\faCaretRight] Else, it outputs ``$\perp$''. If it outputs ``$\perp$'', then $\|\offdiag(\bSigma_T)\|_{\fr,k^2}$ is at most $2\kappa + 2 \rho k$.
\end{itemizec}
Moreover, the algorithm runs in $O\Big(m + sd^{0.62} + d^{1.62+ 3 \frac{\log(4/\rho)}{\log(1/3 \tau)}}  \Big)\poly(n,\log d,1/\tau)$ time.
\end{restatable}
\begin{algorithm}[t]
\caption{Main Subroutine}
\label{alg:win-win-algorithm-app}
\begin{algorithmic}[1]
\Require Frobenius threshold $\kappa \in \R_+$, a finite set $T \subset \R^d$ such that diagonal entries of $\Sigma_T$ lie in $[1/2,2]$, correlation threshold $\rho \in (0,1)$, weak correlation threshold $\tau$, 
sampling parameter $m \in \N$. 
We require the parameters to satisfy 
\begin{align}
\label{eq:constrain-on-m-s}
\frac{c_1d^2}{m}\left(\frac{\kappa^2}{\tau^2} + \log d\right) \leq s.
\end{align}
\Ensure With high probability, output either (i) a set $H \subset [d] $ with $|H| \leq k''$ and $\|(\offdiag(\bSigma_T)_{H}\|_\fr \geq 0.5\kappa$ or (ii) ``$\perp$''; 
	If it outputs ``$\perp$'', then $\|(\offdiag(\bSigma_T)_{H}\|_\fr \leq 2\kappa + 2k \rho$.
\State Let $H'_{\text{pair}}$ be the output of Random Subsampling Algorithm (\Cref{alg:randomlycheck}) with $m$ and $\tau$
\If{$|H'_{\text{pair}}| \geq \kappa^2/\tau^2$}\label{line:block-with-randomsubsample-start-app}
\State Set $H_1 \gets \proj_{\kappa^2/\tau^2}(H_{\text{proj}})$
\State \Return $H_1$
\label{line:block-with-randomsubsample-end-app}
\Else 
\label{line:else-label-H_pair-large-app}
				\State Let $H_{\text{pair}} \subseteq [d] \times [d]$ be the output of  \Cref{thm:robust-correlation-detection} with $\tau, \rho, s$ on $T$ 
		\If {$|H_{\text{pair}}| > \kappa^2/\rho^2$}
				\State $H_3 \gets \proj_{\kappa^2/\rho^2}(H_{\text{proj}})$  
		\State \Return $H_3$
		\Else 
			\State Let $H_2 \gets \proj(H_{\text{pair}})$
			\If{$\|(\offdiag(\bSigma_S))_{H_2}\|_\fr \geq \kappa$}
			\label{line:FronNormH_2-app}
			\State \Return $H_2$
			\Else 
			\State \Return ``$\perp$''
			\EndIf
		\EndIf  
\EndIf
\end{algorithmic}
\end{algorithm}

\begin{proof}
Let $\bA$ be the matrix with diagonal entries zero and off-diagonal entry $(i,j)$ equal to $\frac{\bSigma'_{i,j}}{ \sqrt{\bSigma'_{i,i}\bSigma'_{j,j}}}$.
Since diagonal entries of $\bSigma_T$ lie in $[0.5,1.5]$, the entries of $\bA$ and $\offdiag(\bSigma_T) = \offdiag(\bSigma_T - \bI_d)$ have the same signs and have the magnitude up to a factor of $2$.
Thus, at the cost of constant factor, we will upper bound $\|\bA\|_{\fr,k^2}$ and lower bound $\|(\bA)_H\|_\fr$ instead of dealing with $\offdiag(\bSigma_T - \bI_d)$. 

We will use the notation from \Cref{alg:win-win-algorithm-app}.
Let $s_*$ be the number of off-diagonal coordinates of $\bA_*$ that have absolute value larger than $\tau$.
Our algorithm will be correct on the event when both \Cref{thm:robust-correlation-detection} and \Cref{lem:randomly-check-coordinates} succeed, which happens with high probability.
For the rest of this proof, we condition on both of these algorithms succeeding.

We first consider the case when $s_* > s$.
Then, the lower bound on $s$ in the statement (cf.\ \Cref{eq:constrain-on-m-s}), coupled with  \Cref{lem:randomly-check-coordinates}, implies that 
\begin{align*}
|H_{\text{pair}}'| &> \frac{ms_*}{4d^2} - 16\log d > \frac{ms}{4d^2} - 16\log d \\
&> \frac{m}{4d^2} \left( \frac{c_1d^2}{m} \left( \frac{\kappa^2}{\tau^2} + \log d \right) \right) - 16\log d > \frac{c_1\kappa^2}{2\tau^2}, 
\end{align*}
where we use $c_1$ is large enough, say $ c_1 \geq 100$.
By definition of $H_1 = \proj_{\kappa^2/\tau^2}(H_{\text{proj}})$, there are at least $|\kappa^2/\tau^2|$ entries in $(\bA)_{H_1}$ with absolute value at least $\tau$, and thus $(\bA)_{H_1}$ has Frobenius norm at least $\sqrt{\kappa^2/\tau^2}\tau = \kappa$.

Consider the alternate case when $s_* < s$.
If $|H_{\text{pair}}'|$ happens to be large, then the same argument as above implies the correctness and the runtime.
Thus, in the rest of this proof, we consider the case when we enter Line \ref{line:else-label-H_pair-large-app}.
Since $s_* \leq s$, 
\Cref{thm:robust-correlation-detection} runs in the desired time and finds all off-diagonal coordinate pairs of $A$, collected in $H_{\text{pair}}$, with entries bigger than $\rho$ in the promised time.
If $H_{\text{pair}}$ has more than $\kappa^2/\rho^2$ entries, then by definition of $H_3 = \proj_{\kappa^2/\rho^2}(H_{\text{pair}})$, the same argument as above implies that the Frobenius norm of $(\bA)_{H_3}$ is at least $\kappa$.

If $|H_{\text{pair}}| < \kappa^2 /\rho^2$, then $H_2$ is defined to contain all coordinates that appear in $H_{\text{pair}}$.
If Line \ref{line:FronNormH_2-app} succeeds, then the algorithm returns a subset of coordinates satisfying the desired conditions (small size and Frobenius norm larger than $\kappa$).
Suppose the if condition is not satisfied (and we return ``$\perp$'').
We argue that $\|\bA\|_{\fr,k^2}$ is small enough. 
Write $\bA = \bB +  \bB'$, where $\bB_{i,j}$ is non-zero only when $i\in H_2, j \in H_2$ with values equal to $\bA_{i,j}$ and zero otherwise ($\bB'$ is then defined to be $\bA-\bB$).
By definition of $H_2$, we know each entry of $\bB'$ is at most $\rho$ in absolute value.
By triangle inequality,
 $\|\bA\|_{\fr,k^2} \leq \|\bB\|_{\fr,k^2} + \|\bB'\|_{\fr,k^2} = \|(\bA)_{H_2}\|_{\fr,k^2} + \|\bB'\|_{\fr,k^2} \leq \|(\bA)_{H_2}\|_\fr + k \rho \leq \kappa + k\rho$.
 The accuracy guarantee follows by noting that the entries of $\bA$ and $\bSigma-\bI$ are within a factor of $2$.
Finally, the runtime guarantees are immediate by \Cref{thm:robust-correlation-detection,lem:randomly-check-coordinates}.
\end{proof}

\subsection{Proof of \Cref{thm:formal-main-result}}
\label{app:proof_of_cref_thm_formal_main_result}

The complete algorithm is given below:
\begin{algorithm}[H]
\caption{Main Algorithm}
\label{alg:wrapper-algorithm}
\begin{algorithmic}[1]
\Require corruption rate $\epsilon \in (0,1)$, stability parameter $\delta \in (0,1)$, corrupted set $T \subset \R^{d}$, correlation-threshold $\rho \in (0,1)$, 
 correlation-decay $q \in \N$, sparsity $k \in \N$, sampling parameter $m \in \N$. We require $T$  to be an $\epsilon$-corrupted version of an $(C\epsilon, \delta, k')$ stable set with respect to $\mu$ for $k':= \frac{(Ck)^q}{(\delta^2/\eps)^{q-1}}$ and $\delta^2/\eps \leq c$ for a small absolute constant $c>0$.
\Ensure $\widehat{\mu} \in \R^d$ such that, with high probability, $\|\widehat{\mu} - \mu\|_{2,k} \lesssim \delta$. 
\State $T' \gets $ Filter $T$ using \Cref{claim:preprocessing}
\Statex\Comment{Preprocessing along the diagonals to ensure \Cref{cond:idealistic-condition}}
\State $i \gets 1$
\State $T_i \gets T'$.
\State $H \gets$ output of \Cref{prop:win-win-guarantee}  with inputs: corrupted set $T_i$, Frobenius threshold $ \kappa = C' \delta^2/\eps$ for a large constant $C'$, correlation threshold $\rho = (\delta^2/\eps)/k$, margin threshold $\tau = (\rho/12)^q$, sampling parameter $m \asymp d(\kappa^2/\tau^2 + \log d)$, and correlation count $s= d$
\label{line:choice-of-parameters}
\Statex
	\Comment{\Cref{alg:win-win-algorithm-app}}
	\While{$T_i \neq \emptyset$ and $H \neq$ ``$\perp$''}
	\State Get the scores $f: T_i \to \R_+$ from \Cref{lem:sparse-filtering-lemma} with inputs $T_i$, $H$, $\epsilon$, and  $\delta$ 
	\State $T_{i+1} \gets $ Filter $T_{i}$ using the scores $f$ similar to \Cref{alg:randomized_filtering}
	\State $i \gets i+1$
	\State Update $H$ as above
	\EndWhile
	\State \Return the sample mean of $T_i$
\end{algorithmic}
\end{algorithm}
We now present the proof of its correctness:
\begin{proof}[Proof of \Cref{thm:formal-main-result}]
The first step of \Cref{alg:wrapper-algorithm} is the fast preprocessing step from \Cref{claim:preprocessing}, which takes at most $\tilde{O}(dk^2|T|^2)$ time and removes not too many inliers with high probability
In particular, the returned set $T'$ has diagonal values in $[1/2,2]$.
In fact, \Cref{lem:inherit-stability-mean} implies that the subsequent sets $T_i$'s will satisfy this property, at least until we remove more than $\Omega(\eps)$ fraction of points, which is the regime of interest anyway (cf.\ \Cref{thm:generic-randomized-alg}). 
Thus, \Cref{prop:win-win-guarantee}
 will be applicable.

We will follow the standard proof template of randomized  filtering algorithm from \Cref{thm:generic-randomized-alg}, with the stopping condition provided by \Cref{prop:win-win-guarantee}.
The choice of parameters in Line~\ref{line:choice-of-parameters} are such that whenever \Cref{prop:win-win-guarantee} does not return ``$\perp$'' (and $T_i$ is $10\eps$-corruption of $S$), it returns a set $H$ satisfying the guarantees of \Cref{lem:sparse-filtering-lemma}.
To see this, observe that whenever \Cref{alg:win-win-algorithm-app} outputs a subset $H \subset [d]$, then $|H| \leq \kappa^2/\tau^2$ and $\|(\bSigma_{T_i} - \bI_d )_{H}\|_\fr > \|(\offdiag(\bSigma_{T_i}))_{H}\|_\fr > \kappa/4 \gtrsim  \delta^2/\eps$.
Moreover, whenever \Cref{prop:win-win-guarantee} returns ``$\perp$'', then $\|(\bSigma_{T_i} - \bI_d )_{H}\|_{\fr,k^2} \lesssim \kappa + k \rho \lesssim \delta^2/\eps$ since $\kappa \asymp \delta^2/\eps$ and $\rho\asymp \delta^2/\eps$.
The choice of sparsity parameter in the stability of inliers is $k'$, which is larger than $\kappa^2/\tau^2$, because the choice of $\kappa$ above and $\tau = (\rho/12)^q$ imply that $\kappa^2/\tau^2$  is of the order $(\delta^2/\eps)/(c'\rho^q) \asymp (\delta^2/\eps)/( ((c'\delta^2/\eps)/k)^q) \asymp \frac{k^q}{C^q (\delta^2/\eps)^{q-1}}$.
Thus, the scores generated by $H$ using \Cref{lem:sparse-filtering-lemma} give more weights to outliers than inliers.
\Cref{thm:generic-randomized-alg} now implies that the final set $T''$, with probability at least $0.6$, satisfies (i) $T''$ is an $O(\eps)$-contamination of $S$, (ii) $T''$ is an $O(\eps)$ contamination of $T'$, where $T'$ satisfies \Cref{cond:idealistic-condition}, and (iii) $\|\offdiag(\bSigma_{T''} - \bI_d )_{H}\|_{\fr,k^2} \lesssim \delta^2/\eps$.

We now argue that $\|\bSigma_{T''} - \bI_d\|_{\op,k} \lesssim \delta^2/\eps$, and not just the off-diagonal terms.
We follow the following inequalities using triangle inequality:
\begin{align*}
    \left\|\bSigma_{T''} - \bI_d\right\|_{\op,k} &\leq     \left\|\diag\left(\bSigma_{T''} - \bI_d\right)\right\|_{\op,k}
+ \left\|\offdiag\left(\bSigma_{T''} - \bI_d\right)\right\|_{\op,k}\\
&\lesssim \delta^2/\eps + \left\|\offdiag\left(\bSigma_{T''} - \bI_d\right)\right\|_{\fr,k^2} \tag*{(using \Cref{lem:inherit-stability-mean} and $T'$ satisfies \Cref{cond:idealistic-condition})}\\
&\lesssim \delta^2/\eps + \delta^2/\eps \lesssim \delta^2/\eps\,.
\end{align*}
Applying \Cref{lem:certificate-lemma}, the final output of the algorithm, sample mean of $T''$ will be $O(\delta)$ close to the true mean $\mu$, implying the correctness of the procedure.

It remains to show the choice of the parameters $m$, $s$, and $q$ lead to fast runtimes.  
We take $\tau = (\rho/12)^q$ and  $s = d$.
Finally, we take $m \asymp (d^2/s)\cdot(\kappa^2/\tau^2 + \log d) \leq (d \log d)(\kappa^2/\tau^2)$, which satisfies the parameter constraints in \Cref{prop:win-win-guarantee} (cf.\  \Cref{eq:constrain-on-m-s}).
Letting $n=|T|$, the resulting runtime of a single application of \Cref{prop:win-win-guarantee} is thus at most
 \begin{align*}
 A &= \Big(m + sd^{0.62} + d^{1.62+ 3 \frac{\log(4/\rho)}{\log(1/3 \tau)}}  \Big)\poly(n,\log d,1/\tau) \\
 &\leq \Big(d^{1.62+ 3 \frac{\log(4/\rho)}{\log((\rho/4)^q)}}  \Big)\poly(n,\log d,1/\rho^q) \\
 &\leq \Big(d^{1.62+ \frac{3}{q}}  \Big)\poly(n,\log d, k^q, 1/\epsilon^q)\,.
 \end{align*}
Since there are at most $n$ iterations, the claim on the total runtime follows.
Finally, the success probability of the algorithm can be boosted from $0.55$ to $1-\delta$ by repeating the algorithm $\ell = \log(1/\delta)$ times and outputting the estimate that is $O(\delta)$-close (in the $\|\cdot\|_{2,k}$ norm) to the majority of the $\ell$ estimates. It adds only a multiplicative factor of $\ell$ and an additive term of $d \ell^2$ to the runtime.   
\end{proof}

\subsection{Proof of \Cref{thm:informal-main-result}}
\label{app:proof_of_cref_thm_informal_main_result}
We now explain how \Cref{thm:formal-main-result} implies \Cref{thm:informal-main-result}.
\begin{proof}[Proof of \Cref{thm:informal-main-result}]
First, observe that by applying  \Cref{prop:sparse-estimation-using-2k-norm}, the estimation guarantee of \Cref{thm:formal-main-result} can be translated from  $\|\cdot\|_{2,k}$ norm  into $\|\cdot\|_{2}$ norm by hard-thresholding the estimator.
We now turn our attention to the sample complexity. 
Let $S$ be a set of $n$ i.i.d.\ samples from $\cN(\mu,I)$. If $n \geq ((k')^2\log d)/\eps^2$,
then  \Cref{lemma:stab-sample-complexity} implies that $S$ is $(\eps,\delta, k')$ stable with respect to $\mu$ for $\delta \lesssim \eps \sqrt{\log(1/\eps)}$.
Plugging in the value of $\delta$ and $k' \leq \frac{(O(k))^{q}}{\eps^{q-1}}$ in \cref{thm:formal-main-result} yields a sample complexity of at most 
$n \lesssim \frac{(k')^2\log d}{\eps^2} \lesssim (O(k/\eps))^{2q} \log d$. 
The claim on the runtime in \Cref{thm:informal-main-result}  follows from \Cref{thm:formal-main-result} along with the bound on the sample complexity.
\end{proof}

\section{Robust Sparse PCA}
\label{app:robust-sparse-pca}
In this section, we show that the ideas from fast correlation detection can also be useful for robust sparse PCA.
Our main result in this section is the result below:
\ThmInformalRobustSparsePCA*

The result above provides the first subquadratic time algorithm for robust sparse PCA that has error independent of $d$ and $k$.
Similar to the literature on robust sparse mean estimation, existing algorithms  for robust sparse PCA with similar dimension-independent error took $\Omega(d^2)$ time.
However, the error guarantee of \Cref{thm:informal-robust-sparse-pca} is not optimal: the error guarantee above is $(\eps \log(1/\eps)) /\spike$, same as \cite{CheDKGGS21}, as opposed to the near-optimal error of $(\eps^2 \polylog(1/\eps))/\spike$ in \cite{BalDLS17,DiaKKPS19}.

We remark that \Cref{thm:informal-robust-sparse-pca}  does not follow directly from \Cref{thm:informal-main-result}.  In particular, a standard reduction relies on the fact that the mean of the random variable $\bY - \bI_d $ is exactly $\spike vv^\top$ for $\bY = xx^\top$, where $x \sim \cN(0, \bI + \spike vv^\top)$, i.e., (robust) sparse PCA reduces to (robust) sparse mean estimation.
However, $\bY$ is a $d^2$-dimensional object, and thus a naive application of \Cref{thm:informal-main-result} would yield a super-quadratic runtime, in fact,  $(d^2)^{1.62} = \Omega(d^3)$.
Moreover, the covariance of $\bY$ is not isotropic and thus it is unclear if samples from $Y$ would satisfy  the stability condition from \Cref{def:stability-inliers}.
In the rest of this section, we show that a more whitebox analysis of the main ideas from \Cref{thm:informal-main-result}, in particular, \Cref{prop:win-win-guarantee} can yield a subquadratic runtime.

\paragraph{Organization}
This section is organized as follows: We list the deterministic conditions for robust PCA in \Cref{app:deterministic-condition-pca}.
\Cref{app:filters-certificates-and-dense-estimation} contains the results pertaining to the certificate lemma, filtering, and dense estimation algorithm for sparse PCA.
Finally, the algorithm and its proof under a generic stability condition is given in \Cref{app:subquadratic-time-algorithm-for-sparse-pca}.
Finally, we prove \Cref{thm:informal-robust-sparse-pca} in \Cref{sub:proof-of-informal-pca}.

\subsection{Deterministic Conditions}
\label{app:deterministic-condition-pca}
For a set $T$, we use the notation $\overline{\bSigma}_T$ to denote the second moment matrix $\E_{T}[xx^\top]$---not to be confused with the covariance matrix $\bSigma_T$.
\begin{definition}[Stability Condition for PCA]
\label{def:stability-inliers-pca}
For the contamination rate $\epsilon \in (0,1/2)$, error parameter $ \gamma \geq \epsilon$, spike strength $\spike \in \R_+$, and sparsity $k \in \N$, we say a set $S \subset \R^d$ is $(\epsilon, \gamma,k, \spike)$-pca-stable with respect to a $k$-sparse unit vector $v \in \R^d$ and spike strength $\spike$ if
\begin{enumerate}
	\item  For any subset $S' \subseteq S$ with $|S'| \geq (1 -\epsilon)|S|$:$\left\|\overline{\bSigma}_{S'} - (\bI_d + \spike vv^\top)\right\|_{\fr,k^2} {\leq} \gamma$
 
\item For all subsets $H \subset [d]$ with $|H| \leq k$, the set
$\{(x)_H: x \in S\}$ is $(\epsilon, \gamma)$-covariance stable in the sense of \cite[Definition 4.5]{DiaKan22-book} (with respected to an appropriate flattening of $\bI_{|H|})$.
\end{enumerate}
\end{definition}
The second stability condition allows us to perform covariance estimation in Frobenius norm (and hence stronger than principal component analysis) but only when the support is known.

The following result lists the sample complexity of \Cref{def:stability-inliers-pca}.
\begin{lemma}[Sample Complexity]
\label{lem:samp-comp-pca-stab}
Let $S$ be a set of $n$ i.i.d.\ samples from $\cN(0, \bI + \spike vv^\top)$ for $\spike\in (0,1)$ and a $k$-sparse unit vector $v$.
Then if $n \gtrsim \poly(k,\log(d/\delta),1/\epsilon)$, then $S$ is $(\epsilon, \gamma, k, \spike)$-pca-stable with respect to $v$ and $\spike$ for $\gamma = O(\epsilon \log(1/\epsilon))$.
\end{lemma}
\begin{proof}
We sketch the argument here.
The first condition is identical to \cite[Definition 2.2]{CheDKGGS21} and \cite[Lemma 4.2]{CheDKGGS21} establishes a sample complexity of $\frac{k^2 \log(d/ \delta)}{\epsilon^2}$.

For the second condition, we do an admittedly loose analysis. Fixing a subset $H$, it amounts to showing a stability condition for a $|H|$-dimensional Gaussian distribution, for which \cite[Proposition 4.2]{DiaKan22-book} establishes a sample complexity of $\poly(k \log (1/\delta)/\epsilon)$ for the set to be $(\epsilon,\gamma)$-covariance stable for $\gamma = O(\epsilon \log(1/\eps))$ with probability $1- \delta$.\footnote{Although the proof of \cite[Proposition 4.2]{DiaKan22-book} does not explicitly write the dependence on $\delta$, but it is immediate from the proof.}
Since there are at most $d^k$ possible choices of $H$, a union bound shows that the second condition in \Cref{def:stability-inliers-pca} holds with probability $1 - \delta$ with sample complexity $\poly(k \log (d^k/\delta)/\epsilon) = \poly(k \log (d/\delta)/\epsilon)$.
\end{proof}

\subsection{Filters, Certificates, and Dense Estimation}
\label{app:filters-certificates-and-dense-estimation}

\paragraph{Dense Estimation Algorithm}
We shall use the following algorithm to estimate the spike when we have identified the support of the spike.
\label{app:dense-estimation-algorithm}
\begin{lemma}[Dense Covariance Estimation Algorithm; Implication of ~{\cite[Theorem 4.6]{DiaKan22-book}}]
\label{lem:dense-pca-filter}
Let $\epsilon \in (0 ,\epsilon_0)$ for a small absolute constant $\epsilon_0$.
Let $T$ be an $\epsilon$-corrupted version of $S$, where $S$ is $(\epsilon,\gamma,k, \spike)$-pca-stable with respect to $v$ (cf.\ \Cref{def:stability-inliers-pca}) and $\spike \in (0,1)$.
Let $H$ be a $k$-sparse subset of $[d]$.
There exists an algorithm $\cA$ that takes as inputs corrupted set $T$, contamination rate $\epsilon$, error parameter $\gamma$, spike $\spike$, and the sparse support set $H$ and outputs a $k$-sparse vector $u$, supported on $H$, such that  $\left\|\left( \bI_d + \spike uu^\top \right)_H - \left( \bI_d + \spike vv^\top \right)_H \right\| \lesssim \gamma $.
 Moreover, the algorithm runs in time $d\cdot\poly(k,|T|/ \epsilon)$.
 \end{lemma}
\begin{proof}
Let $T' \subset \R^{k}$ and $S' \subset \R^k$ denote the projections of $T$ and $S$, respectively, on $H$.
Let $\bSigma' = (\bI_d + \spike vv^\top)_H = \bI_{|H|} + \spike'  v' {v'}^\top$ for $v' = (v)_H/ \|(v)_H\|_2$ and $\spike' = \spike \|(v)_H\|_2^2$.

\cite[Theorem 4.6]{DiaKan22-book} gives an estimate $ \widehat{\bSigma}$ such that $\|\widehat{\bSigma} - \bSigma'\|_\fr \lesssim\gamma \|\bSigma'\|_\op  \lesssim\gamma$, which can be calculated in time $d\poly(k|T|/\epsilon)$; here $\|\cdot\|_\op$ denotes the operator norm of a matrix.
That is, $ \| (\widehat{\bSigma} -\bI) - \spike' v'{v'}^\top \|_\fr \lesssim\gamma $.
Letting $\bA$ be the best one-rank approximation of $(\widehat{\bSigma} -\bI)$ in the Frobenius norm, we see that
\begin{align*}
\|\bA - \spike'v'{v'}^\top\|_\fr &\leq \|\bA - (\widehat{\bSigma} -\bI)\|_\fr + \|(\widehat{\bSigma} -\bI) - \spike'v'{v'}^\top\|_\fr \leq 2 \|(\widehat{\bSigma} -\bI) - \spike'v'{v'}^\top\|_\fr \lesssim \gamma,
\end{align*}
where the second inequality follows from the fact that $\bA$ is the best rank-one approximation.
Moreover, the symmetry of $\bA$ implies that $\bA$ must be of the form $\lambda' u'{u'}^\top$ for a unit vector $u'$ and $\lambda' \in \R$.
We thus obtain 
\begin{align*}
\|\lambda'u'{u'}^\top - \spike' v'{v'}^\top\|_\fr \lesssim\gamma\,.
\end{align*}
Consequently, $|\lambda' - \spike'| \lesssim\gamma$ by Weyl's inequality.
If $\lambda' \leq 0$, we set $u = 0$.
Otherwise, we set $u = \sqrt{\lambda'/\spike}u'$. 
We now calculate the approximation error. If $\lambda' \geq 0$, then the resulting error in our estimate is $\left\|\left( \bI_d + \spike uu^\top \right)_H - \left( \bI_d + \spike vv^\top \right)_H \right\|_\fr =\left\| \lambda' uu^\top - \spike'v'{v'}^\top \right\|_\fr \lesssim\gamma\,$ by the approximation guarantee.
If $\lambda' < 0$, then $\spike'$ must also be $O(\gamma)$ because $\spike' \leq \lambda' + |\lambda' - \spike'| \lesssim \gamma$.
In this case, the approximation error is $\eta'$, which is also $O(\gamma)$ as required. 
\end{proof}

\paragraph{Certificate Lemma for Sparse PCA}
\label{app:certificate-lemma-for-sparse-pca}
The following result shows that if the covariance matrix looks roughly isotropic on a subset of coordinates $H^\complement$, then the spike vector $v$ places most of its mass on the complement, $H$. 
The benefit of this stopping condition is that it does not depend on the unknown spike vector $v$.
\begin{lemma}[Sparse Certificate Lemma for PCA]
\label{lem:sparse-certificate-pca}
Let $T$ be an $\epsilon$-corrupted version of $S$, where $S$ is $(\epsilon,\gamma,k, \spike)$-pca-stable with respect to a $k$-sparse unit vector $v \in \R^d$ and spike strength $\spike \in (0,1)$.
Let $H \subset [d]$ be such that 
$\| (\overline{\bSigma}_T - \bI)_{H^\complement}\|_{\op,k} = O(\gamma)$,
then $\|vv^\top - (v)_{H}(v)_{H}^\top\|_\fr^2 = O\left( ( \gamma/\spike) +  (\gamma/\spike)^2\right)$.
\end{lemma}
\begin{proof}
Let $z$ be the unit vector along $(v)_{H^\complement}$, which is at most $k$-sparse because $v$ is $k$-sparse.
Since $z$ is supported on $H^\complement$, 
the assumption on $\overline{\bSigma}_T$ and $H^\complement$ implies that 
$\left|z^\top \left( \overline{\bSigma}_T - \bI \right)z \right| \lesssim\gamma$,
which further gives us that $z^\top\overline{\bSigma}_T z = 1 \pm O(\gamma) $.
Using the stability of $S$, we see $z^\top\overline{\bSigma}_T z \geq (1 - \epsilon) z^\top\overline{\bSigma}_{S \cap T} z \geq (1 - \epsilon) \left(  z^\top \left( \bI + \spike vv^\top \right) z  - \gamma\right)$.
Defining $\eta':= \eta(v^\top z)^2 = \eta \|v_{H^\complement}\|_2^2$, we have that 
$ z^\top\overline{\bSigma}_T z \geq (1 - \epsilon)(1 + \spike' - \gamma) $.
Combining this with the aforementioned upper bounds on $z^\top\overline{\bSigma}_T z$, we obtain that
$(1 - \epsilon)(1 + \spike' - \gamma)  \leq 1 + O(\gamma)$.
Thus, $\spike' = O(\gamma + \epsilon) = O(\gamma)$ by using $\eps \leq \gamma$.
Therefore, $\|(v)_{H^\complement}\|_2^2 = \spike'/\spike \lesssim \gamma/\spike$.
Finally, the triangle inequality implies that 
\begin{align*}
\|vv^\top -(v)_H (v)_H\|_\fr \leq \|(v)_{H^\complement}\|_2^2 + 2\|(v)_{H^\complement}\|_2  \lesssim \max\left( \sqrt{(\gamma/\spike)}, (\gamma/\spike) \right)\,.
\end{align*}

\end{proof}

\paragraph{Filter for Sparse PCA}
\label{app:filter-for-sparse-pca}
We use the following result to filter outliers when we identify a small set of coordinates with a large variance.
\begin{restatable}[Sparse PCA Filter]{lemma}{lemSparsePCAFilter}
\label{lem:sparse-pca-filter}
Let $\epsilon \in (0 ,\epsilon_0)$ for a small absolute constant $\epsilon_0$ and let $C$ be a large enough absolute constant.
Let $T$ be an $\epsilon$-corrupted version of $S$, where $S$ is $(\epsilon,\gamma,k, \spike)$-stable with respect to $v$ (cf.\ \Cref{def:stability-inliers-pca}) and $\spike \in (0,1)$.
Let $H \subset [d]$ be such $\left\|(\bSigma_T - \bI - \spike vv^\top)_H\right\|_\fr = \lambda$ for $\lambda \geq 8C\gamma$ and $|H| \leq k$.

Then there exists an algorithm $\cA$ that takes $T$, $H$, $\epsilon$,  $\gamma$, and $\spike$ and returns a score mapping $f:T \to \R_+$ such that
 the sum of inliers' scores is less than outliers': $\sum_{x \in S \cap T} f(x) \leq \sum_{x \in T\setminus S } f(x)$ and $\max_{x \in T}f(x) > 0$.
 Moreover, the algorithm runs in time $d\poly(k,|T|)$.
\end{restatable}
We give the proof of the result above in \Cref{sec:proof-of-lem-sparse-pca-filter}.

\paragraph{Preprocessing Condition for Sparse PCA}
\label{app:idealistic-condition-sparse-PCA}
Similar to robust sparse mean estimation, our algorithm tracks the contribution to the diagonal terms and off-diagonal terms separately. 
The following conditions mirrors \Cref{cond:idealistic-condition}.
\begin{condition}[Idealistic Condition for PCA]
\label{cond:idealistic-condition-pca}
Let $T$ be an $\epsilon$-corrupted version of $S$, where $S$ is $(\epsilon,\gamma,k, \spike)$-stable with respect to a $k$-sparse unit vector $v$ and spike strength $\spike$ (cf.\ \Cref{def:stability-inliers-pca}).
We have an $H_1 \subset [d]$ and $|H_1| \lesssim k^2$ such that $T$ satisfies  $ \|\diag(\overline{\bSigma}_T - \bI_d)_{H_1^\complement}\|_{\fr,k^2} \lesssim  \gamma$.
\end{condition}
The next result gives an efficient algorithm to ensure the above condition:
\begin{restatable}{claim}{ClaimPreprocessingPCA}
\label{claim:preprocessing-pca}
Let $\epsilon \in (0,\epsilon_0)$
and $\gamma \in (0,\gamma_0)$ for small constants $\epsilon_0 \in (0,1/2)$, $\gamma_0 \in (0,1)$.
Let sparsity $k \in \N$.
Let $C$ be a large enough constant and $T$ be an $\epsilon$-corrupted set $S$ where $S$ is $(C\epsilon,\gamma,k',\spike)$-pca-stable with respect to an unknown $k$-sparse unit vector $v \in \R^d$,  $\spike \in (0,1)$, $\gamma \geq \eps$, and $k' = C'k^2$ for a large enough constant $C' >0$.
There is a randomized algorithm $\cA$ that takes as input the corrupted set $T$, contamination rate $\epsilon$, sparsity $k \in \N$, 
and a parameter $q \in \N$, and returns a set $T' \subset T$ and $H_1 \subset [d]$ in time $O(d \poly(k,|T|))$ such that with probability $0.9$  
\begin{enumerate}
\item $T'$ is an $O(\eps)$-contamination of $S$.
    \item Each diagonal entry of $\overline{\bSigma}_{T'} \in [1/2,4]$.
\item $T'$ and $H_1$ satisfy  \Cref{cond:idealistic-condition-pca}, i.e., $|H_1| \lesssim k^2$ and $ \|\diag(\overline{\bSigma}_{T'} - \bI_d )_{H_1^\complement}\|_{\fr,k^2} \lesssim  \gamma$.    
\end{enumerate}

\end{restatable}

The proof of this simple claim is given in \Cref{app:proof-of-claim-preprocessing-pca}.

Similar to \Cref{lem:inherit-stability-mean}, the next result shows that all large subsets of a set satisfying \Cref{claim:preprocessing-pca} are close to identity in the sparse operator norm.
\begin{restatable}{claim}{ClInheritStabilityPCa}
\label{lem:inherit-stability-pca}
Let $C$ be a large enough constant $C>0$ and $k,k' \in \N$.
 Let $T'' \subset T'$ be two $O(\eps)$-contamination of $S$ such that $S$ is an $(C\eps,\gamma,k',\spike)$-pca-stable with respect to $v$ and $\spike \in (0,1)$.
 Let $H \subset [d]$ be a small subset $|H| \leq k$ such that
 $ \|\diag(\overline{\bSigma}_{T'} - \bI_d )_{H_1^\complement}\|_{\fr,k^2} \lesssim \gamma$.
Then
 $ \|\diag(\overline{\bSigma}_{T''} - \bI_d )_{H_1^\complement}\|_{\op,k} \lesssim \gamma$. 
\end{restatable}

\subsection{Subquadratic Time Algorithm For Sparse PCA}
\label{app:subquadratic-time-algorithm-for-sparse-pca}

We now establish the main technical result of this section:

\begin{theorem}[Subquadratic Time Algorithm for Robust Sparse PCA under Stability]
\label{thm:robust-sparse-pca-stability}
Let $\epsilon \in (0,\epsilon_0)$
and $\gamma \in (0,\gamma_0)$ for small constants $\epsilon_0 \in (0,1/2)$, $\gamma_0 \in (0,1)$.
Let $C$ be a large enough constant and  $k \in \N$ be the sparsity parameter and $q \in \N$ the correlation decay parameter with $q \geq 3$.
Let $T$ be an $\epsilon$-corrupted set $S$ where $S$ is $(C\epsilon,\gamma,Ck',\spike)$-pca-stable with respect to an unknown $k$-sparse unit vector $v \in \R^d$,  $\spike \in (0,1)$, and $\gamma \geq \eps$
for $k' = \left( \frac{ C^q k^{2q}}{ \gamma^{2q-2}}\right)$.

There is a randomized algorithm $\cA$ that takes as input the corrupted set $T$, contamination rate $\epsilon$, stability parameter $\gamma$, sparsity $k \in \N$, 
and a parameter $q \in \N$, and produces an estimate $\widehat{v}$ such that
\begin{itemizec}
	\item (Error)  We have $\|\widehat{v}\widehat{v}^\top - vv^\top\|_\fr \lesssim 
	 \sqrt{\frac{\gamma}{\spike}} $ with high probability over the randomness of the samples and the algorithm.
	\item (Runtime) The algorithm runs in time at most
$d^{1.62 + \frac{3}{q}} \poly\left( |T| \log d, (k/\epsilon)^q   \right).$
\end{itemizec}
\end{theorem}
The complete algorithm is given in \Cref{alg:wrapper-algorithm-pca}.
\begin{algorithm}
\caption{Robust Sparse PCA Algorithm}
\label{alg:wrapper-algorithm-pca}
\begin{algorithmic}[1]
\Require corruption rate $\epsilon \in (0,1)$, stability parameter $\gamma \in (0,1)$, corrupted set $T \subset \R^{d}$, correlation-threshold $\rho \in (0,1)$, 
 correlation decay $q \in \N$, sparsity $k \in \N$,
 spike strength $\eta \in (0,1)$, sampling parameter $m \in \N$. We require $T$  to be an $\epsilon$-corrupted version of an $(C\epsilon, \delta, Ck',\spike)$ stable set with respect to $\mu$ for $k':= \left( \frac{ C^q k^{2q}}{ \gamma^{2q-2}}\right)$ and $\gamma \ll 1$.
\Ensure $\widehat{v} \in \R^d$ such that, with high probability, $\|\widehat{v}\widehat{v} - vv^\top\|_{\fr} \lesssim \gamma$. 
\State $H_1,  T' \gets $ be the outputs of \Cref{claim:preprocessing-pca} on $T$
\Statex \Comment{$T'$ is an $O(\eps)$-contamination of $S$, $ 0.5 \bI_d \preceq \diag(\overline{\bSigma}) \preceq 2\bI_d $, $|H_1| \lesssim k^2$ and $ \|\diag(\overline{\bSigma}_T - \bI_d )_{H_1^\complement}\|_{\fr,k^2} \lesssim  \sqrt{\gamma/\eta}$ }
\State $i \gets 1$
\State $T_i \gets T'$.
\State $H \gets$ output of \Cref{prop:win-win-guarantee}  with inputs: corrupted set $T_i$, Frobenius threshold $ \kappa = \Theta(\gamma)$, correlation threshold $\rho = \gamma/k$, margin threshold $\tau = (\rho/12)^q$, sampling parameter $m \asymp d(\kappa^2/\tau^2 + \log d)$, and correlation count $s= d$
\label{line:choice-of-parameters-pca}
\Comment{\Cref{alg:win-win-algorithm-app}}

	\While{$T_i \neq \emptyset$ and $H \neq$ ``$\perp$'' and $|H|\leq \kappa^2/\rho^2$}
	\State Get the scores $f: T_i \to \R_+$ from \Cref{lem:sparse-pca-filter} with inputs $T_i$, $H$, $\epsilon$, and  $\delta$ 
	\State $T_{i+1} \gets $ Filter $T_{i}$ using the scores $f$. 
	\State $i \gets i+1$
	\State Update $H$ as above
	\EndWhile
\State  Let $H_{\text{end}} \gets H \cup H_1$
	\State $u \gets $ output of \Cref{lem:dense-pca-filter} with inputs $T_i$, $\eps$, $\gamma$, $\spike$, and $H_{\text{end}}$
\end{algorithmic}
\end{algorithm}
We now present the proof of its correctness:
\begin{proof}[Proof of \Cref{thm:robust-sparse-pca-stability}]
We can assume that $\gamma \leq \spike$, otherwise we can simply output any unit vector, whose error will be at most $O(1) = O(\sqrt{\gamma/\spike})$. 

We will use the same template of filtering-based algorithms from \Cref{thm:generic-randomized-alg}, with the filtering subroutines provided by \Cref{lem:sparse-pca-filter}.

The first step of \Cref{alg:wrapper-algorithm-pca} is the preprocessing step from \Cref{claim:preprocessing-pca}, which takes at most $\tilde{O}(d\poly(k,|T|))$ time and removes not too many inliers with high probability.
 In particular, the diagonal entries of $\overline{\bSigma}_{T'}$ lie in $[1/2,4]$; Moreover, all the subsequent sets $T_i$'s will continue to satisfy this property by \Cref{lem:inherit-stability-pca} as long as we have removed at most $O(\eps)$ fraction of points (since $\gamma$ is small enough). 
 We shall use the output $H_1$ generated by \Cref{claim:preprocessing-pca} in the end.

We will use a fine-grained result from \Cref{prop:win-win-guarantee}.
Observe that \Cref{prop:win-win-guarantee} can be used not just for the covariance $\bSigma_T$ but also for $\overline{\bSigma}_T$, i.e., without centering, which is what we will use in this proof. 
We shall invoke \Cref{prop:win-win-guarantee} to identify all the coordinate pairs with correlation larger than $\rho = \gamma/k$.
We then set $\tau = (\rho/12)^q$
and use the Frobenius threshold $\kappa = C' \gamma$ for a constant $C>100$.
Defining $k' := \kappa^2/\tau^2$ and $k'':= \kappa^2/\rho^2 $, we note that $k' \geq k'' = C' \gamma^2/ (\gamma/k)^2 =C' k^2 \geq 2k^2$. 
We will show that the set $H$ allows us to filter points using \Cref{lem:dense-pca-filter}.
The proof of \Cref{prop:win-win-guarantee} reveals that it returns $H$ such that either
\begin{enumerate}[label=(Case \Roman*), labelindent=10em ,labelwidth=1.3cm, labelsep*=1em,leftmargin=1.9cm]
 	\item $|H| = k' = \kappa^2/ \tau^2$ and for each coordinate in $i \in H$, there exists a $j \in H$ such that $(i,j)$ is $\tau$-correlated  
    
    Let $T''$ be the current iterate of the corrupted data set.
	Since $\spike vv^\top$ is a $k^2$-sparse matrix and $\offdiag\left(\overline{\bSigma}_T - \bI_d\right)$ has at least $k'$ entries with magnitude at least $\Theta(\tau)$\footnote{Recall that each diagonal entry of $\overline{\bSigma}_{T''}$ is $\Theta(1)$ and thus $\tau$-correlation implies that the corresponding entry in $\overline{\bSigma}_{T''}$ is also $\Theta(\tau)$.}, their difference must have at least $k' - k^2 \geq k'/2 = \kappa^2/\tau^2$ entries with magnitude $\Theta(\tau)$.
	Thus, the sparse Frobenius norm of their difference must be large, i.e.,
	\begin{align*}
	\left\| \left(\overline{\bSigma}_{T''} - \bI_d - \spike vv^\top\right)_H\right\|_{\fr}  \gtrsim \kappa \gtrsim \gamma\,.
	\end{align*}
	Given such an $H$, we filter outliers using \Cref{lem:dense-pca-filter}.

 	\item $|H| \geq \kappa^2/ \rho^2$ and for each coordinate in $i \in H$, there exists a $j \in H$ such that $(i,j)$ is $\rho$-correlated  

 	By the same argument as above, the matrix $\overline{\bSigma}_{T''} - \bI_d - \spike vv^\top$ has at least $k'' - k^2 \geq k''/2 =\kappa^2/(2 \rho^2)$ entries with magnitude $\Theta(\rho)$.
 	Thus, $\left\| \left(\overline{\bSigma}_{T''} - \bI_d - \spike vv^\top\right)_H\right\|_{\fr} \gtrsim \kappa \gtrsim \gamma$.
	Again, we filter outliers using \Cref{lem:dense-pca-filter}.

 	\item $|H| \leq \kappa^2/\rho^2 $ and no coordinate pair in $H^\complement$ is $\rho$-correlated.

  Let $T''$ be the current iterate of the corrupted data set.
    Since $H \subset H_{\text{end}} $ and the coordinates in $H^\complement$ are at most $\rho$-correlated, it follows that the coordinates in $H_{\text{end}}^\complement$ are also at most $\rho$-correlated. 
 	\begin{align}
 	\left\| \offdiag \left(\left(  \overline{\bSigma}_{T''}  - \bI_d \right)_{H_{\text{end}}^\complement}\right) \right\|_{\fr,k^2} \leq \left\| \offdiag\left(\left(  \overline{\bSigma}_{T''}  - \bI_d \right)_{H^\complement}\right) \right\|_{\fr,k^2} \lesssim \rho k \lesssim \gamma\,.
  \label{eq:sparse-pca-proof-offdiag}
 	\end{align}
    We now want to combine the above guarantee on the closeness along the offdiagonals with the guarantee on the closeness along the diagonals from $H_1$.  
    In particular, $H_{\text{end}}$ satisfies that  
    \begin{align}
        \left\|\diag\left(\left(\overline{\bSigma}_{T''} - \bI_d\right)_{H_{\text{end}}^\complement}\right)\right\|_{\op,k} \leq \left\|\diag\left(\left(\overline{\bSigma}_{T''} - \bI_d\right)_{H_{1}^\complement}\right)\right\|_{\op,k} \lesssim \gamma\,,
  \label{eq:sparse-pca-proof-diag}
    \end{align}  
    where $T''$ has small sparse operator norm because it close to the preprocessed set $T'$ and thus \Cref{lem:inherit-stability-pca} is applicable.
    Combining \Cref{eq:sparse-pca-proof-offdiag} and \Cref{eq:sparse-pca-proof-diag}, we obtain
    \begin{align}
    \left\|(\overline{\bSigma}_{T''} - \bI_d)_{H_{\text{end}}^\complement}\right\|_{\op,k} 
    &\leq      
    \left\|\diag\left(\left(\overline{\bSigma}_{T''} - \bI_d\right)_{H_{\text{end}}^\complement}\right)\right\|_{\op,k} + 
    \left\|\offdiag\left( \left(\overline{\bSigma}_{T''} - \bI_d\right)_{H_{\text{end}}^\complement}\right)\right\|_{\op,k} 
    \lesssim\gamma  \,.
    \label{eq:pca-operatorn-norm-bound-proof}
    \end{align}
    The size of $H_{\text{end}}$ is at most $O(k^2) + k' \leq 2k'$.
    Since $S$ satisfies stability with $2k'$, \Cref{lem:sparse-certificate-pca} then implies that the spike $v$ is mostly contained in $H_{\text{end}}$.
 	By invoking the dense PCA algorithm on $H_{\text{end}}$, \Cref{lem:dense-pca-filter} estimates the spike $(v)_{H_{\text{end}}}(v)_{H_{\text{end}}}^\top$ with $uu^\top$, with $u$ supported on $H_{\text{end}}$. Combining, we obtain: 
 	\begin{align*}
 	\left\|vv^\top - uu^\top\right\|_\fr &\leq \left\|vv^\top - (v)_{H_{\text{end}}}(v)_{H_{\text{end}}}^\top\right\|_\fr + \left\|uu^\top - (v)_{H_{\text{end}}}(v)_{H_{\text{end}}}^\top\right\|_\fr \\
 	&\lesssim \sqrt{\gamma/\eta}  +  \gamma/\spike \lesssim \sqrt{\gamma/\spike}, 
 	\end{align*}
 	where the first term is bounded using \Cref{lem:sparse-certificate-pca} with \Cref{eq:pca-operatorn-norm-bound-proof} and the second term is bounded using \Cref{lem:dense-pca-filter}.

 	\item $H = $ ``$\perp$''.

 	The same argument as the previous case holds.

 \end{enumerate} 
Thus, the output of the algorithm is $O( \sqrt{\gamma/\spike})$ close as required.
It remains to show the choice of the parameters $m$, $s$, and $q$ lead to the claimed runtime.  
We take $\tau = (\rho/12)^q$ and  $s = d$.
Finally, we take $m \asymp (d^2/s)\cdot(\kappa^2/\tau^2 + \log d) \leq (d \log d)(\kappa^2/\tau^2)$, which satisfies the parameter constraints in \Cref{prop:win-win-guarantee} (cf.\  \Cref{eq:constrain-on-m-s}).
Letting $n=|T|$, the resulting runtime of a single application of \Cref{prop:win-win-guarantee} is thus at most
 \begin{align*}
 A &= \Big(m + sd^{0.62} + d^{1.62+ 3 \frac{\log(4/\rho)}{\log(1/3 \tau)}}  \Big)\poly(n,\log d,1/\tau) \\
 &\leq \Big(d^{1.62+ 3 \frac{\log(4/\rho)}{\log((\rho/4)^q)}}  \Big)\poly(n,\log d,1/\rho^q) \\
 &\leq \Big(d^{1.62+ \frac{3}{q}}  \Big)\poly(n,\log d, k^q, 1/\epsilon^q)\,.
 \end{align*}
As the iteration count is bounded by $n$, the total runtime is at most $nA$.
\end{proof}

\subsection{Proof of \Cref{thm:informal-robust-sparse-pca}}
\label{sub:proof-of-informal-pca}

\begin{proof}[Proof of \Cref{thm:informal-robust-sparse-pca} using \Cref{thm:robust-sparse-pca-stability}]
    By \Cref{thm:robust-sparse-pca-stability}, it suffices to establish that a set of $n$ i.i.d.\ samples from $\cN(0,\bI + \spike vv^\top)$ for a $k$-sparse unit vector $v$, with high probability, is $(\eps, \gamma, k',\spike)$-pca-stable for $\gamma \lesssim \eps \log(1/\eps)$ and $k':= (\frac{C^qk^{2q}}{\gamma^{2q-2}})$, where $C$ is a large absolute constant.
    \Cref{lem:samp-comp-pca-stab} gives a bound on the sample complexity, stating that it suffices to take $n = \poly(k', \log(d), 1/\eps)$ many samples, thus establishing \Cref{thm:informal-robust-sparse-pca}.
\end{proof}

\section{Discussion}
\label{sec:discussion}
In this article, we presented the first subquadratic time algorithm for robust sparse mean estimation. We now discuss some related open problems and avenues for improvement.
First, the sample complexity of \Cref{thm:informal-main-result} is polynomially larger than $k^2 (\log d)/\epsilon^2$---the sample complexity of existing (quadratic runtime) algorithms.\footnote{The sample complexity of $\Tilde{\Theta}(k^2/\epsilon^2)$ is also conjectured to be near-optimal among computationally-efficient algorithms~\cite{BreBre20}.}
Bridging this gap is an important problem to improve sample efficiency.
Second, \Cref{thm:informal-main-result} is specific to isotropic distributions whose quadratic polynomials have bounded variance  (distributions $P$ with mean $\mu$ that satisfy $\Var_{x \sim P}((x-\mu)^\top \bA (x - \mu)) \lesssim \|\bA\|_\fr^2$ for all symmetric matrices $\bA$).
Indeed, both \cite{DiaKKPS19,CheDKGGS21} rely on the isotropy and the aforementioned variance structure of quadratic polynomials to avoid solving SDPs that appear in \cite{BalDLS17}. 
To the best of our knowledge,
even obtaining an $O(d^2)$ runtime algorithm for unstructured distributions is still open.
Third, because \Cref{thm:informal-main-result} relies on \cite{Valiant15}, which in turn relies on fast matrix multiplication, the resulting algorithm may not offer practical benefits for moderate dimensions; see the discussion in \cite{Valiant15}.
We believe overcoming these limitations is an important practically-motivated question.
Finally,
\Cref{ques:linear} remains open.

\section*{Acknowledgements}
We are grateful to Sushrut Karmalkar, Jasper Lee, Thanasis Pittas, and Kevin Tian for  discussions on robust sparse mean estimation.
We also thank Ilias Diakonikolas and Daniel Kane for their many insights on robustness.

\printbibliography

\newpage
\appendix

\section{Details Deferred from \Cref{sec:preliminaries_for_sparse_estimation}}
\label{app:preliminaries}
In this section, we include details that were omitted from \Cref{sec:preliminaries_for_sparse_estimation}.
\Cref{app:proof_of_sparse_filtering_lemma} provides the proof of \Cref{lem:sparse-filtering-lemma}.
\Cref{app:preprocessing} ensures the preprocessing condition listed in \Cref{cond:idealistic-condition}.
Finally, \Cref{app:fast_correlation_detection} gives further details about the correlation detection algorithm from \cite{Valiant15}.

\subsection{Defining Sparse Scores}
\label{app:proof_of_sparse_filtering_lemma}
In this section, we give the proof of \Cref{lem:sparse-filtering-lemma}.
\LemSparseFilteringLemma*
\begin{proof}
As a first step, we simply project the data points along the coordinates in $H$, and by abusing the notation, call the projected set $T$.
In the remainder of this proof, $\bI$ refers to $\bI_{|H|}$.
Computing this projected set takes at most $d|T||H|$ time.

Let $\lambda = \|\bSigma _{T} - \bI\|_\fr$. Define $\bA = (\bSigma_{T} - \bI)/\|\bSigma_{T_1} - \bI\|_{\fr}$ so that $\bA$ maximizes the trace inner product with $\bSigma_T-\bI$ over unit Frobenius norm matrices.
Define the function $g(x):= (x - \mu_{T})^\top \bA (x - \mu_{T}) - \trace(\bA)$.
We first compute the average of $g(x)$ over the $T$ below:
\begin{align*}
\E_T[g(x)] = \E_{T}\left[ \left\langle (x - \mu_T)(x- \mu_T)^\top - \bI , \bA \right \rangle  \right] = \langle \bSigma_{T} - \bI , \bA \rangle = \lambda\,.
\end{align*}
By abusing notation again, we use $S$ to denote the projection of the inliers $S$ along the coordinates in $H$;
Note that the projected set also inherits $(\epsilon,\delta,k)$-stability and the $\|\cdot\|_{\fr,k^2}$ reduces to the standard Frobenius norm.  
Thus, for any large subset of $S'\subset S$ with $|S'| \geq (1 -\epsilon)|S|$, we use \Cref{lem:certificate-lemma} to obtain the following:
\begin{align*}
\left|\E_{S'}[g(x)]\right| &= \left|\E_{S'}\left[\left\langle (x - \mu_{T})(x - \mu_{T})^\top - \bI, \bA\right\rangle\right]\right|\\
&= \left|\langle \bSigma_{S'} - \bI , \bA \rangle + 2 (\mu - \mu_T)^\top \bA (\mu - \mu_{S'}) + (\mu - \mu_{T})^\top \bA  (\mu - \mu_{T})\right| \\
&\leq \|\bSigma_{S'} - \bI\|_\fr + 2 \|\mu - \mu_T\|_2 \|\bA\|_\fr \|\mu - \mu_{S'}\|_2 + \|\mu- \mu_T\|_2^2 \|\bA\|_\fr^2\\
&\lesssim \delta^2/\epsilon + 2 (\delta + \sqrt{\epsilon \lambda}) \delta + (\delta + \sqrt{\epsilon \lambda 	})^2 \tag*{(using stability and \Cref{lem:certificate-lemma})}  \\
&\lesssim \delta^2/\epsilon + 3\delta^2 + 4\delta \sqrt{\epsilon \lambda}  + \epsilon \lambda  \\
&\lesssim \delta^2/\epsilon + 7\delta^2 + 2\epsilon \lambda  \tag*{(using  $2ab\leq a^2 + b^2$)}\\
&\lesssim \delta^2/\epsilon + \epsilon \lambda,
\numberthis \label{eq:guuarantee-S'}
\end{align*}
where we used that $\epsilon \leq 1$. The following helper result, which is a slight generalization of \cite[Proposition 2.19]{DiaKan22-book} from non-negative $h$'s to real-valued $h$'s  will be useful.
\begin{claim}
\label{cl:scores-over-inliers}
Let $h:S \to \R$ be a real-valued function on a finite set $S$.
Further suppose that $\left|\E_{S'}[h(X)]\right| \leq \tau$ for all sets $S' \subset S$ with $|S'| \geq (1- \eps)|S|$. 
For an $\eps \leq 1/2$,  defining $f'(x) = h(x)\mathbf{1}_{h(x) \geq 3 \tau/\eps}$, we have that $\E_{S}[f'(x)] \leq 3\tau$. 
\end{claim}
\begin{proof}
     First, we show the following:
    for all subsets $S'' \subset S$ with $|S''| \leq \eps|S|$,
\begin{align}
  \E_{S''}[\max(h(X),0)] \leq \tau + 2\tau/\eps
  \label{eq:sum-inliers-part-1}\,.
\end{align} 
    To that end, 
    for all sets $S''$ with $|S''| = \eps |S|$, the triangle inequality implies 
    \begin{align}
    \label{eq:sum-over-small-filter}
        \left|\E_{S''}[h(X)] \right|=\left| \frac{\E_S[h(X)] - (1-\eps)\E_{S \setminus S''}[h(X)]}{\eps} \right| \leq 2 \tau/\eps.
    \end{align}
    Let $S_*$ be the top $\eps |S|$ entries of $S$ in the increasing order of $h(\cdot)$, not in the absolute value. Then establishing \Cref{eq:sum-inliers-part-1} is equivalent to establish an upper bound on
    $\frac{1}{|S_*|} \sum_{x \in S_*} \max(h(X),0)$.    %
    Thus, if all the entries of $S_*$ are bigger than $0$, then \Cref{eq:sum-inliers-part-1} follows by \Cref{eq:sum-over-small-filter}.
    We shall show that $h(x)$ on $S_*$ is lower bounded by $-\tau$. Under this condition, we see that the desired result follows similarly  by \Cref{eq:sum-over-small-filter}: $\frac{1}{|S_*|} \sum_{x \in S_*} \max(h(X),0) \leq \frac{1}{|S_*|} \sum_{x \in S_*} \left(h(X) + \tau\right) \leq \tau + 2\tau/\eps$.
We now establish that $\min_{x \in S_*} h(x) \geq -\tau$.    If there exists an $x \in S_*$ with $h(X) \leq -\tau$, then the average of $S\setminus S_*$ must be less than $-\tau$, contradicting the assumption that $|\E_{S\setminus S_*} [h(X)]| \leq \tau$. Thus, we have established \Cref{eq:sum-inliers-part-1}.

Given \Cref{eq:sum-inliers-part-1}, we see that the fraction of points with $h(x) \geq 3 \tau/\eps$ must be less than $\eps |S|$. Otherwise, the conditional average over those points would be at least than $3\tau/\eps \geq 2 \tau/\eps + \tau$, contradicting \Cref{eq:sum-inliers-part-1}.
Therefore, the function $f$ is non-zero only on at most $\eps$-fraction of $S$. Therefore, the non-negativity of $f$ implies that
\begin{align}
    \frac{1}{|S|} \sum_{x \in S} f(X) \leq \frac{1}{|S|}\max_{S'' \subset S: |S''| \leq \eps |S|  } \sum_{x \in S''} f(X) \leq  \frac{1}{|S|}\max_{S'' \subset S: |S''| \leq \eps |S|  } \sum_{x \in S''} \max(h(X), 0) \leq \eps( \tau + 2 \tau/\eps ),
\end{align}
where we use non-negativity of $f$ and \Cref{eq:sum-inliers-part-1}.
\end{proof}
Let $R:= C'(\frac{\delta^2}{\eps} + \eps \lambda)$, for a large constant $C'>0$, be the bound from  \Cref{eq:guuarantee-S'}. Combining this with the claim above, we see that defining $f(x)$ to be $g(x)\mathbf{1}_{x \geq 3R/\eps}$, the sum of scores over inliers is small:

\begin{align*}
\sum_{x \in S \cap T} f(x) \leq \sum_{x \in S} f(x) \leq 3R|S| = 3C'(\delta^2/\eps + \eps \lambda)|S| \leq 0.25 \lambda|S|  \,,    
\numberthis\label{eq:sum-over-inliers}
\end{align*}
where we use that $ 3C'\epsilon \leq 1/8$ and $3C'\delta^2/\epsilon \leq \lambda/8$.

On the other hand, $\E_{x \in T\setminus S}[f(x)]$ must be large argued as argued below. We observe that $f(x) \geq g(x) - 3R/\eps$, and thus applying \Cref{eq:guuarantee-S'} to $T \cap S$, which is of size at least $(1 - \epsilon)|T| = (1 - \epsilon)|S|$, we obtain
\begin{align*}
\sum_{x \in T\setminus S}f(x) &\geq  \sum_{x \in  T \setminus S}  (g(x) - 3R/\eps) =\left(\left(\sum_{x \in T} g(x) \right) -  \left(\sum_{x \in T \cap S} g(x)\right) \right) -  \left(|T \setminus S| 3R/\eps \right)\\
&\geq \left(  \lambda |T|\right) - \left((|T \cap S|) R\right)  - \epsilon |T| \left( 3 R/\eps \right) \\
&\geq \lambda |T| - 4|T| R \geq \lambda|T|/2\,,
\numberthis\label{eq:sum-scores-T}
\end{align*}
where the last inequality follows if we show that $R \leq  \lambda/8$.
Indeed, this follows if $C'\delta^2/\eps \leq \lambda/16$ and $C'\eps \leq 1/16$.
Combining \Cref{eq:sum-over-inliers} and \Cref{eq:sum-scores-T}, we get the desired result; the claim on $\max_{x \in T}f(x)$ follows from \Cref{eq:sum-scores-T}.
The complete  algorithm is given below:
\begin{algorithm}
\caption{Quadratic Scores}
\begin{algorithmic}[1]
\State $T \gets \{(x)_H: x \in T\}$ \Comment{Projection onto $H$} 
\State Let $\bA$ be the matrix $(\bSigma_{T} - \bI)/\|\bSigma_{T} - \bI\|_{\fr}$ \Comment{The matrix such that  $\langle \bA , \bSigma_T - \bI \rangle = \|\bSigma_{T} - \bI\|_\fr$}
\State Define $g(x):= (x - \mu_{T})^\top \bA (x - \mu_{T}) - \trace(\bA)$
\State Define $f(x)$ to be $g(x)$ if $g(x) \geq 3C'\left(\frac{\delta^2}{\epsilon^2} +  \lambda\right)$ otherwise $0$
\State Return $f$
\end{algorithmic}
\end{algorithm}

\end{proof}

\subsection{Preprocessing: Proofs of \Cref{claim:preprocessing,lem:inherit-stability-mean}}
\label{app:preprocessing}

In this section, we outline how to ensure \Cref{cond:idealistic-condition} quickly using \Cref{alg:randomized_filtering,lem:sparse-filtering-lemma}.
\CondDataProcess*
Given the sparse filtering lemma (\Cref{lem:sparse-filtering-lemma}),
we can simply filter along the diagonals to ensure \Cref{cond:idealistic-condition} as shown below.
\begin{claim}
\label{claim:preprocessing}
Let $\epsilon \in (0 ,\epsilon_0)$ for a small absolute constant $\epsilon_0$.
Let $c$ be a small enough absolute constant and $C$ be a large enough constant.
Let $T$ be an $\epsilon$-corrupted version of $S$, where $S$ is $(\epsilon,\delta,k^2)$-stable with respect to $\mu$ such that $\delta^2/\epsilon \leq c$.
Then there is a randomized algorithm $\cA$ that takes as inputs $\epsilon$, $\delta$, and $k$
such that it outputs a set $T' \subseteq T$ such that with probability at least $8/9$:
(i) $T'$ is at most $10\epsilon$-corruption of $S$ and (ii) $\|\diag\left( \bSigma_{T'} - \bI_d \right)\|_{\fr,k^2} \leq  C\delta^2/\epsilon \leq  0.1$, and (iii) the algorithm runs in time $\tilde{O}(|T|dk^2 + |T|^2d)$. 
\end{claim}
\begin{proof}
We can simply invoke \Cref{alg:randomized_filtering} with the stopping condition on the set $T_i$ set to $$\|\diag\left( \bSigma_{T_i} - \bI_d \right)\|_{\fr,k} \leq C\delta^2/\epsilon$$ for a constant $C$ large enough.
To evaluate this stopping condition, we can simply compute the matrix $\diag\left( \bSigma_{T'} - \bI_d \right)$ in $O(d|T|)$ time.
 The associated $\|\cdot\|_{\fr,k^2}$ can be easily calculated by computing the Euclidean norm of its largest $k^2$ entries, again computable in $\tilde{O}(dk^2)$ time.

If the stopping condition is not satisfied, \Cref{lem:sparse-filtering-lemma} returns the required scores.
Thus, we get the desired guarantees on the set $T$ from \Cref{thm:generic-randomized-alg}.
\end{proof}
We now give the proof of \Cref{lem:inherit-stability-mean}.
\LemInheritStability*
\begin{proof}
    First we note that the lower bound on the sparse eigenvalues follow rather directly as shown below. 
    We make use of the equality $\bSigma_{T''} := \frac{1}{|T''|^2}\sum_{x,y \in T''} (x - y)(x-y)^\top$.
    For any sparse unit vector $v$, we use the stability condition applied to $S \cap T''$ to obtain the following:
    \begin{align*}
 v^\top \bSigma_{T''} v &=         \frac{1}{|T''|^2} \sum_{x, y \in T''} (v^\top (x-y))^2 \geq  \frac{1}{|T''|^2} \sum_{x, y \in S \cap T''} (v^\top (x-y))^2 \tag*{(using non-negativity)}\\
 &= \frac{|S\cap T''|^2}{|T''|^2} v^\top \bSigma_{S \cap T''} v  \geq (1 - O(\eps))v^\top \bSigma_{S \cap T''} v\\
 &= (1 - O(\eps)) \left( v^\top \left( \E_{S \cap T''} \left[ (X - \mu)(X - \mu)^\top\right] \right) v -  \left(v^\top (\mu - \mu_{S \cap T''})\right)^2 \right) \\
 &\geq (1 - O(\eps)) \left( 1 - O(\delta^2/\eps) -O(\delta^2) \right) \tag*{(using stability)}\\
 &\geq 1 - O(\delta^2/\eps), \numberthis \label{eq:idealistic-persists-lower-bound}
    \end{align*}
where the last inequality uses $\eps \leq \delta$.
For the upper bound, we observe that for any matrix $\bA$, $\|\diag(\bA)\|_{\op,k}$ is attained by $1$-sparse unit vectors $v$, i.e., $\|\bA\|_{\op,1} = \|\diag(\bA)\|_{\op,k}$.
Thus, for any $T'' \subset T'$ and that $|T'| \leq |T''|(1 + O(\eps))$:
for any $1$-sparse unit vector $v$, 
    \begin{align*}
 v^\top \bSigma_{T''} v &=         \frac{1}{|T''|^2} \sum_{x, y \in T''} (v^\top (x-y))^2 \leq  \frac{|T'|^22}{|T''|^2} \frac{1}{|T'|^2} \sum_{x, y \in T'} (v^\top (x-y))^2 \tag*{(using nonnegativity)}\\
 &= \frac{|T'|^22}{|T''|^2} v^\top \bSigma_{T'} v \leq (1 + O(\eps))v^\top \bSigma_{T'} v  \leq (1 + O(\eps))(1 + O(\delta^2/\eps)) = 1 + O(\delta^2/\eps)\,,
\numberthis \label{eq:idealistic-persists-upper-bound}
    \end{align*}
    where we use that for $1$-sparse unit vectors $v$, $v^\top \bSigma_{T'} v = v^\top \diag\left(\bSigma_{T'}\right) v$.
Combining \Cref{eq:idealistic-persists-lower-bound,eq:idealistic-persists-upper-bound} for all $1$-sparse unit vectors $v$, we obtain that  $\|\diag\left(\bSigma_{T''} - \bI\right)\|_{\op,k} = O(\delta^2/\eps)$.

\end{proof}

\subsection{Fast Correlation Detection}
\label{app:fast_correlation_detection}
In this section, we show how to obtain \Cref{thm:robust-correlation-detection} from \cite[Theorem 2.1]{Valiant15}.
\begin{theorem}[{Robust Correlation Detection For Boolean Vectors in Subquadratic Time~\cite[Theorem 2.1]{Valiant15}}]
\label{thm:binary-robust-correlation-detection}

Consider a set of $n'$ vectors in $\{-1, 1\}^{d'}$ and constants $\rho, \tau \in [0,1]$ with $\rho > \tau$ such that for all but at most $s$ pairs $u, v$ of distinct vectors, $|u^\top v|/ \|u\|_2 \|v\|_2 \leq \tau$.
 There is an algorithm that, with probability $1-o(1)$, will output all pairs of vectors whose normalized inner product is least $\rho$.
 Additionally, the runtime of the algorithm is 
 \begin{align*}
 \left(sd'{n'}^{0.62} + {n'}^{1.62 + 2.4 \frac{\log(1/\rho)}{\log(1/\tau)}}\right) \poly\left( \log n, 1/\tau \right)\,.
\end{align*}
\end{theorem}
An improved algorithm with better runtime was then provided in \cite[Corollary 1.8]{KarKK18-correlation}, but we choose the version above for its simplicity.
We provide the proof of \Cref{thm:robust-correlation-detection}, the version we used in  this work, from \Cref{thm:binary-robust-correlation-detection} using standard arguments below:
\begin{proof}[Proof of \Cref{thm:robust-correlation-detection} from {\cite[Theorem 2.1]{Valiant15}}]
Let $X \in \R^{d \times n}$ denote the matrix with the columns of $X$ denoting the centered vectors of $T$.
That is, if $T = \{z_1,\dots,z_n\}$, then the $i$-th column of $X$ is equal to $z_i - \mu_T$.
 Let $X_i$ denote the $i$-th row of the matrix $X$.
For $i,j \in [d] \times [d]$, the correlation between the $i$-th and the $j$-th coordinate on $T$, 
$\corr(i,j)$,
is equal to $\left|\frac{X_i^\top X_j}{ \|X_i\|_2 \|X_j\|_2}\right|$.
Thus, we would like to apply \Cref{thm:binary-robust-correlation-detection} with the rows of $X$ (thus $n' = d$ and $d' = n$).
However, $X'$ is not a binary matrix. 

A standard reduction allows us to compute a binary matrix that preserves the correlation between the rows of $X$. 
Let $G \in \R^{n \times m}$ be a matrix with independent $\cM(0,1)$ entries.
Let $Y =  X G$ be in $\R^{d \times m}$ and let $Y' = \sign(Y) \in \R^{d \times m}$, where $\sign$ is applied elementwise.
Thus, $Y'$ is a boolean matrix as required in \Cref{thm:binary-robust-correlation-detection}.
The following arguments show that the correlation between the rows of $Y$ is preserved for $m$ large enough.
Let $Y_i,Y'_i$ denote the $i$-th row of the matrices $Y, Y'$.
\begin{lemma}[{\cite[Lemma 4.1]{Valiant15}}]
If $m \geq 10 \log(n)/ \gamma^2$, then
with probability $1 -o(1)$,
we have that for all $i \neq j \in [n]$,
we have that 
\begin{align*}
\left| \frac{\langle Y'_i, Y'_j \rangle}{\|Y'_i\|_2 \|Y'_j\|_2} - \frac{2}{\pi} \arcsin\left( \frac{\langle X_i, X_j \rangle}{\|X_i\|_2 \|X_j\|_2} \right)\right| \leq \gamma.
\end{align*}
\end{lemma}
In particular, if the original correlation is less than  $\tau$ in the absolute value,
then the corresponding correlation in $Y'$ in the absolute value is at most $(2/ \pi) \arcsin(\tau) + \gamma \leq (4/ \pi) \tau + \gamma \leq 2 \tau + \gamma$, where we use $|\arcsin(x)| \leq |2x|$.
Similarly, if the original correlation is at least $\rho$ in absolute value, then the new correlation  is at least $(2/ \pi) \arcsin(\rho) - \gamma \geq (2 / \pi) \rho - \gamma \geq \rho/2 - \gamma  $.
Choosing $\gamma = \tau$, the new matrix $Y'$ satisfies the guarantees of \Cref{thm:binary-robust-correlation-detection} with $n' = d, d' = m = 10 \log(n)/\tau^2, \tau' = 3 \tau , \rho' = \rho/4$ (using $ \tau \leq \rho/4$).
The time taken to compute $Y$ and $Y'$ is at most $n dm$.
Thus, the total runtime is at most
\begin{align*}
&\left( nd \frac{\log n}{\gamma^2} + s \frac{\log n}{\tau^2} d^{0.62} + d^{{1.62} + 2.4 \frac{\log(4/\rho)}{\log(1/ 3 \tau)}} \poly\left( \log n, 1/ \tau \right)  \right) \\
&\qquad\leq \left(  s d^{0.62} + d^{{1.62} + 2.4 \frac{\log(4/\rho)}{\log(1/ 3 \tau)}}  \right) \cdot \poly\left( n (\log d)/\gamma^2 \right)\,.
\end{align*}
The probability of success can be boosted using repetition, if needed.

\end{proof}

\section{Details Deferred from \Cref{app:robust-sparse-pca}}

In this section, we give the proofs of \Cref{lem:sparse-pca-filter,claim:preprocessing-pca,lem:inherit-stability-pca}.
\subsection{Proof of \Cref{lem:sparse-pca-filter} }
\label{sec:proof-of-lem-sparse-pca-filter}

\lemSparsePCAFilter*
\begin{proof}
We use the same ideas from the proof of  \Cref{lem:sparse-filtering-lemma} and similarly assume that $T$ and $S$ already correspond to the projected coordinates.
The challenge in applying the idea as is lie in the uncertainty about $v$.
We thus use \Cref{lem:dense-pca-filter} to first estimate $v$ using the returned vector $u$ which satisfies that $\|\spike uu^\top - \spike vv^\top \|_\fr \lesssim \gamma$.\footnote{Observe that $v$ here corresponds to $(v)_H$ because of the projection to $H$. Hence $v$ is no longer a unit vector.}
We give the algorithm below. 
\begin{algorithm}
\caption{PCA Filter}
\begin{algorithmic}[1]
\State $T \gets \{(x)_H: x \in T\}$ \Comment{Projection onto $H$} 
\State $u \gets$ be output of dense covariance estimation algorithm \Cref{lem:dense-pca-filter}
\State Let $\bA$ be the matrix $(\bSigma_{T} - \bI - \spike uu^\top 	)/\|\bSigma_{T_1} - \bI - \spike uu^\top\|_{\fr}$ 
\State Define $g(x):= (x - \mu_{T})^\top \bA (x - \mu_{T}) - \trace(\bA)$
\State Define $f(x)$ to be $g(x)$ if $g(x) \geq \Omega(\gamma/\epsilon)$ otherwise $0$
\State Return $f$
\end{algorithmic}
\end{algorithm}

Let  $\lambda = \left\|(\bSigma_T - \bI - \spike vv^\top)_H\right\|_\fr$ and define $\bA = (\bSigma_{T} - \bI - \spike uu^\top)/\|\bSigma_{T_1} - \bI- \spike uu^\top\|_{\fr}$.
Define the function $g(x):= x^\top \bA x - \langle \bI + \spike uu^\top, \bA\rangle$.
Computing the average of $g(x)$ over the $T$, we obtain 
\begin{align*}
\E_T[g(x)] &= \E_{T}\left[ \left\langle xx^\top - \bI - \spike uu^\top , \bA \right \rangle  \right] = \langle \overline{\bSigma}_T - \bI - \spike uu^\top , \bA \rangle = 
\| \overline{\bSigma}_T - \bI - \spike uu^\top  \|_\fr\\
&\geq \| \overline{\bSigma}_T - \bI - \spike vv^\top \|_\fr - \|\spike uu^\top - \spike vv^\top \|_\fr = \lambda - O(\gamma) \geq 3 \lambda/4 
\numberthis\label{eq:sum-over-T-of-g-pca},
\end{align*}
where we use that $\lambda \gtrsim \gamma$.
Using the stability of $S$ (observe that on the projected set, the $\|\cdot\|_{\fr,k^2}$ norm becomes the usual Frobenius norm), for any large subset of $S'\subset S$ with $|S'| \geq (1 -\epsilon)|S|$,  the closeness between $u$ and $v$ implies 
the following:
\begin{align*}
\left|\E_{S'}[g(x)]\right| &= \left|\E_{S'}\left[\left\langle xx^\top - \bI - \spike uu^\top, \bA\right\rangle\right]\right| =  \left|\left\langle \overline{\bSigma}_{S'} - \bI - \spike uu^\top , \bA \right\rangle\right| \\
&= \left|\left\langle \overline{\bSigma}_{S'} - \bI - \spike vv^\top, \bA \right\rangle + \left\langle \spike vv^\top - \spike  uu^\top , \bA \right\rangle\right| \leq \left\|\overline{\bSigma}_{S'} - \bI - \spike vv^\top\right\|_\fr + \spike\left\|  vv^\top -   uu^\top \right\|_\fr \\
&\leq \gamma + O(\gamma)  \leq C \gamma
\numberthis\label{eq:sum-of-g-over-S'-pca}
\end{align*}
for a large enough absolute constant $C$.
Now define $f(x) := g(x) \mathbf{1}_{g(x) \geq 3C \gamma/\eps}$, i.e.,
 $f(x)$ is equal to $g(x)$ if $g(x) \geq 3C \gamma/\eps$ and $0$ otherwise. 
 By \Cref{cl:scores-over-inliers}, 
\begin{align}
\sum_{x \in S \cap T} f(x) \leq \sum_{x \in S} f(x) \leq 3 C \gamma|S| \leq |T| \lambda/4\,,
\label{eq:sum-over-inliers-pca}
\end{align}
where we use that $ \lambda \gtrsim C\gamma$.

On the other hand, $\E_{x \in T\setminus S}[f(x)]$ must be large argued as argued below. Applying \Cref{eq:sum-of-g-over-S'-pca} to $T \cap S$, which is of size at least $(1 - \epsilon)|T| = (1 - \epsilon)|S|$, we obtain
\begin{align*}
\sum_{x \in T\setminus S}f(x) &\geq  \sum_{x \in  T \setminus S}  (g(x) - 3C \gamma/\epsilon) =\left(\left(\sum_{x \in T} g(x) \right) -  \left(\sum_{x \in T \cap S} g(x)\right) \right) -  \left(|T \setminus S| 3C \gamma/\epsilon \right)\\
&\geq \left(  3\lambda |T|/4\right) - \left((|T \cap S|) (C\gamma) \right)  - \epsilon |T| \left(3C\gamma/\epsilon \right) \\
&\geq \left(  3\lambda |T|/4\right) - \left(4|T|C \gamma \right) \\
&\geq \lambda|T|/2\,,
\numberthis\label{eq:sum-scores-T-pca}
\end{align*}
where we use \Cref{eq:sum-over-T-of-g-pca} and $\lambda \gtrsim \gamma$.
The desired conclusions follow from \Cref{eq:sum-over-inliers-pca} and \Cref{eq:sum-scores-T-pca}.
\end{proof}

\subsection{Proofs of \Cref{claim:preprocessing-pca,lem:inherit-stability-pca}}
\label{app:proof-of-claim-preprocessing-pca}

\ClaimPreprocessingPCA*    
\begin{proof}

We will filter the set using \Cref{lem:sparse-pca-filter} following the template of \Cref{alg:randomized_filtering} until the second and the third conditions are met.

Starting with the second condition, the lower bound on the diagonal entries follows by the fact that $T$ contains an $\eps$-fraction of $S\cap T$ and diagonal entries of $\overline{\bSigma}_{S \cap T} $ are at least $1 - \gamma$.
Since $\eps$ and $\gamma$ are small enough, $(1-\eps)(1 - \gamma)\geq 1/2$.
We now focus on establishing the upper bound, for each coordinate $i \in [d]$,
the true variance is at most $(1 + \spike\|v\|_\infty^2) \leq 2$, where $\|\cdot\|_\infty$ denotes the $\ell_\infty$ norm.
Thus, if for any coordinate $i \in [d]$, the empirical variance of $T$ is larger than $3$,
which is bigger than $1 + \eta + \lambda$,
$H := \{i\}$ satisfies the condition of \Cref{lem:sparse-pca-filter}.
Thus, we can filter points until all the empirical variances are less than $3$ following the template of \Cref{thm:generic-randomized-alg}.

We now turn our attention to the third condition. For a vector $x\in \R^d$, let the function $g: \R^d \to \R^d$ be the coordinate-wise square of its input.
Thus, for $X \sim \cN(0, \bI + \spike vv^\top)$, we have that  $ \mu:=\E[g(X)] = u + \spike g(v)$, where $u\in\R^d$ denotes the all ones vector.
Therefore, the required condition in \Cref{cond:idealistic-condition-pca} can also be written equivalently as
$\| \left( \mu_{f(T)}- u\right)_{H_1^\complement} \|_{2,k^2} = O(\gamma)$.
Let $f(T)$ and $f(S)$ denote the sets transformed by $f$, i.e.,
  $f(T) := \{f(x): x \in T\}$ and  $f(S) := \{f(x): x \in S\}$. 
We first compute $\mu_{f(T)} - u$, which takes $O(d|T|)$ time.
Let $J \subset [d]$ denote the coordinates $i$ for which the $i$-th coordinate of  $\mu_{f(T)} - u$ is bigger than $\gamma/k$ in absolute value.
Let $J_*$ denote the support of the sparse spike vector $v$.

Consider the case when $J$ is a large enough set: $|J| > k'$. 
Then let $H \subset J$ be any subset of size $k'$.
Then $|H \setminus J_*| > |H| - |J_*| = k' - k \geq k'/2$, And thus we have that 
\begin{align*}
    \|(\bSigma_T - \bI_d - \spike vv^\top)_{H}\|_{\fr,{k'}^2} &\geq
    \|(\bSigma_T - \bI_d - \spike vv^\top)_{H\setminus J_*}\|_{\fr,{k'}^2} \geq \\
 &\geq   \| (\bSigma_T - \bI_d )_{H\setminus J_*}\|_{\fr,{k'}^2} 
 \\
& \geq    \| \diag\left((\bSigma_T - \bI_d )_{H\setminus J_*}\right)\|_{\fr,{k'}^2} \\
&\geq (\gamma/k)\sqrt{(k'/2)} \gtrsim \gamma\,,
\end{align*}
since $k'\gtrsim k^2$.
Thus, we have obtained a set $H$ that satisfies the guarantees of \Cref{lem:sparse-pca-filter}, which allows us to filter using the template of \Cref{alg:randomized_filtering}.

If, on the other hand, $J$ happens to be small, then we return $H_1 = J$,
and since $\mu_{f(T)}$ is $(\gamma/k)$-close to $u$ (that is, all-ones vector) on $H_1^\complement$,
we obtain
\begin{align*}
    \|\diag(\bSigma_T - \bI_d )_{H_1^\complement}\|_{\fr,k^2}  = 
    \| \left( \mu_{f(T)}- u\right)_{H_1^\complement} \|_{2,k^2} \leq k \gamma/k \leq \gamma. \qquad\qquad\qquad\qquad \qedhere
\end{align*}
\end{proof}

\ClInheritStabilityPCa*
\begin{proof}
Our proof strategy will be similar to that of \Cref{lem:inherit-stability-mean}, and we refer the reader to the proof of \Cref{lem:inherit-stability-mean} for more details.
Starting with the lower bound on the sparse eigenvalues of $\overline{\bSigma}_{T''}$, we note that for any $k'$-sparse unit vector $u$, the stability condition applied to $S \cap T''$ implies
    \begin{align*}
 u^\top \overline{\bSigma}_{T''} u &\geq \frac{|S\cap T''|}{|T''|} u^\top \bSigma_{S \cap T''} u \geq (1 - O(\eps)) (u^\top (I + \spike vv^\top) u - \gamma)  \geq (1 - O(\eps))(1 - O(\gamma)) \geq 1 - O(\gamma)
    \end{align*}
since $\eps \leq \gamma$.
Applying this inequality for $1$-sparse unit vectors $u$, we obtain that  $-u^\top \diag(\overline{\bSigma}_T - \bI) u \leq C \gamma $ for a large constant $C > 0$.
Turning towards the upper bound, we proceed as follows: for any $1$-sparse unit vector $u$ supported on $H_1^\complement$,     
\begin{align*}
 u^\top \left(\diag\left(\bSigma_{T''} \right)\right)_{H_1^\complement}u &= u^\top \bSigma_{T''} u =          \frac{1}{|T''|} \sum_{x\in T''} (u^\top x)^2 \leq  (1 + O(\eps)) u^\top \bSigma_{T'} u  \leq (1 + O(\eps))(1 + O(\gamma)) = 1 + O(\gamma)\,.
    \end{align*}
Overall, we obtain the desired guarantee of $\left\|\diag\left(\left(\bSigma_{T''} - \bI\right)_{H_1^\complement}\right)\right\|_{\op,k} = O(\gamma)$.

\end{proof}

\end{document}